\documentclass[letterpaper,USenglish,cleveref,autoref,thm-restate]{lipics/lipics-v2021}
\hideLIPIcs
\nolinenumbers

\usepackage{amsmath}
\usepackage{amssymb}
\usepackage{bbm}
\usepackage{mathtools}
\usepackage{makecell}
\usepackage{natbib}
\usepackage{todonotes}
\usepackage{xfrac}
\usepackage{cleveref}
\usepackage{ifthen}

\graphicspath{ {./figures/} }

\makeatletter
\AtBeginDocument{%
  \@ifundefined{r@prop:demand-functions}{}{%
    \global\expandafter\let\csname r@prop:demand_curve\expandafter\endcsname
      \csname r@prop:demand-functions\endcsname
    \global\expandafter\let\csname r@prop:demand_curve@cref\expandafter\endcsname
      \csname r@prop:demand-functions@cref\endcsname
  }%
}
\makeatother

\newtheorem{assumption}[theorem]{\text{Assumption}}

\newcommand{\E}{\mathbb{E}}

\newcommand{\R}{\mathbb{R}}
\newcommand{\1}{\mathbbm{1}}
\newcommand{\F}{\mathcal{F}}

\newcommand{\Var}{\mathrm{Var}}

\newcommand{\LVR}{\mathrm{LVR}}
\newcommand{\paraheader}[1]{\smallskip\noindent{\sffamily\bfseries #1}}

\provideboolean{submissionversion}
\setboolean{submissionversion}{true}

\ifthenelse{\boolean{submissionversion}}{
  \renewcommand{\todo}[2][1]{}
}{
}

\title{Uniform-Loss Automated Market Making for Prediction Markets} %TODO Please add

\author{Ciamac C.~{Moallemi}}{Decision, Risk, and Operations Division, Graduate School of Business,
  Columbia University
  \and Paradigm
  \and Uniswap Labs}{}{}{CCM is supported by the Briger Family Digital
  Finance Lab at Columbia University. CCM is a research advisor for Paradigm and Uniswap Labs.}

\author{Dan {Robinson}}{Paradigm}{}{}{}

\author{Brian {Zhu}}{Department of Industrial Engineering and Operations Research, Columbia
  University}{}{}{}

\authorrunning{Ciamac C.~Moallemi, Dan Robinson, and Brian Zhu}

\Copyright{Ciamac C.~Moallemi, Dan Robinson, and Brian Zhu}


 %TODO mandatory; please add comma-separated list of keywords

\ArticleNo{0}
\begin{document}

\maketitle

%TODO mandatory: add short abstract of the document
\begin{abstract}
Automated market makers (AMMs) for prediction markets descend from market scoring rules, where a mechanism operator subsidizes a market to aggregate beliefs about uncertain events.
The existing literature has focused on bounding the total worst-case loss to the subsidizer, but has not addressed how that loss is distributed across price states or over time.
We use the framework of loss-versus-rebalancing (LVR) to study this distribution and introduce \textit{uniform AMMs}, defined by the property that instantaneous LVR is proportional to pool value and independent of the current token price.
In a static setting, we show that for a broad class of \textit{win-martingales} --- processes that converge to 0 or 1 at a fixed resolution time --- there exists a pricing function that achieves uniform LVR under that process, and conversely, that any sufficiently regular pricing function induces a win-martingale under which it is uniform.
We then extend the framework to dynamic liquidity management, showing that liquidity levels can be adjusted over time to implement a prescribed target expected cumulative loss schedule.
This theory is illustrated with canonical examples of win-martingales and pricing functions.
Our results can inform AMM designers and liquidity providers on how the inevitable cost of subsidizing price discovery can be shaped and controlled across both price and time.
\end{abstract}

\section{Introduction}

Prediction markets aggregate beliefs about uncertain future events through the trading of
contingent claims whose prices are naturally interpreted as probabilities
\citep{arrow_debreu_1954,wolfers2004prediction}. In a binary prediction market, two outcome
contracts --- commonly called YES and NO tokens --- pay one unit of numeraire depending on whether
a specified event occurs, so that their prices sum to one and the YES price represents the
market-implied probability of the event. Providing liquidity is essential for these markets to
function, yet many prediction markets concern long-tail or short-lived events for which continuous
quoting by traditional market makers is impractical. Market scoring rules, introduced by
\citet{Hanson03,hanson2007lmsr}, address this problem by having a mechanism operator maintain an
automated market maker (AMM) that stands ready to trade at prices determined by a scoring rule or
invariant. In this framework, the operator acts as a subsidizer of price discovery: by committing
capital to the AMM, the operator absorbs losses to informed traders in exchange for the
information revealed by trading activity. The equivalence between these cost-function
prediction-market mechanisms and the constant-function market makers that later emerged in
decentralized finance has been formalized by \citet{frongillo2024axiomatic}.

A central achievement of the market scoring rule literature has been to bound the total worst-case
loss to the subsidizer. In the logarithmic market scoring rule (LMSR), for example, the operator's
maximum loss is $b\ln n$ for $n$ outcomes with liquidity parameter $b$
\citep{hanson2007lmsr}. \citet{chen2007utility} extend this analysis to a broader class of
cost-function market makers, providing a utility-theoretic framework for bounded-loss
mechanisms. These bounds are valuable because they make the subsidy budget transparent and
finite. However, they treat the subsidy as a single aggregate number. They do not address how the
loss is distributed: whether it concentrates at extreme or moderate prices, whether it accumulates
early or late in the market's life, or whether some regions of the probability space are
effectively unsubsidized. The distribution of losses across price states and over time has not
been a focus of prediction market mechanism design.

Loss-versus-rebalancing (LVR), introduced by \citet{milionis2022automated}, provides the
analytical framework to study exactly this distribution. LVR formalizes the instantaneous
adverse-selection cost of an AMM position: it measures the rate at which value flows from the
liquidity pool to arbitrageurs who trade against stale quotes. In prediction markets, LVR has a
natural dual interpretation: it is both the cost borne by the liquidity provider and the
subsidy flowing to informed traders who move stale prices toward the current fair value. Because
LVR is an instantaneous and state-dependent quantity, it reveals how the subsidy is allocated
across price states at each moment in time, providing exactly the granularity that the
bounded-loss framework leaves unexplored.

Outcome contracts in prediction markets stand out from other assets as their prices are bounded
between zero and one, and at resolution they collapse to one of the two endpoints. Before
resolution, prices move as information arrives. Under a risk-neutral measure \citep{harrison1979martingales}, the fair price is therefore a probability-valued martingale: conditional on the current
information available to the market, the expected future price equals the current price. The
binary resolution condition further implies that the terminal value must lie in $\{0,1\}$. These
processes are referred to in the literature as \emph{win-martingales}. They provide a reduced-form
model of belief dynamics in binary prediction markets: the martingale property captures the
intertemporal consistency of prices or beliefs, while terminal absorption at zero or one captures
event resolution. Since volatility is the central driver of adverse selection, volatility
structure of the win-martingale determines where and when the subsidy is consumed, making it the
key input to AMM design.

This paper studies AMM design for binary prediction markets through the lens of \emph{uniform
  LVR}. An AMM is uniform with respect to a given price process if its instantaneous LVR is
proportional to pool value and independent of the current price state, after normalizing by the
rate of information arrival. Uniformity ensures that the subsidy is spread evenly across belief
states: no region of the probability space is over- or under-subsidized relative to pool
value. For the subsidizer, this makes the cost of the mechanism legible --- expected losses can be
expressed as a fraction of pool value rather than as a complicated function of the current
probability. For the mechanism designer, it provides a principled way to allocate liquidity across
belief states by matching the geometry of the AMM to the stochastic dynamics of the event
probability. In this sense, uniform LVR addresses the question of \textit{how} the bounded loss
should be distributed, complementing the classical question of \textit{how much} the total loss
is.

Beyond legibility, the uniform-LVR criterion has several normative justifications that we develop
in \cref{sec:static}. First, uniformity yields a tractable liquidity schedule for implementing any
prescribed target expected loss profile, enabling the dynamic results in
\cref{sec:dynamic}. Second, if liquidity providers can dynamically enter and exit, uniform LVR
ensures that the loss rate per dollar deployed is independent of the current price state,
eliminating adverse selection among LPs. Third, among designs with a given total expected loss,
uniform LVR minimizes the worst-case relative loss rate across price states, a minimax fairness
criterion. These properties distinguish uniform LVR from alternative loss distributions and
motivate it as a principled design criterion rather than merely a convenient simplification.

The key concept we introduce is that uniformity is not a property of an invariant alone, but rather a joint property of an AMM and a price process. We work with separable win-martingales, in which the diffusion coefficient factors into a state-dependent volatility profile and a time-dependent information-release rate. This separability distinguishes where information moves prices most strongly from when information arrives most quickly, and we study a range of volatility profiles and information arrival patterns with natural economic interpretations. On the AMM side, we represent an invariant through its pool-value function, the marked-to-market value of the arbitraged pool as a function of the fair price. Requiring LVR to be proportional to pool value across all prices yields a boundary value problem (BVP) that is the central object of the paper. We show that this BVP gives a bidirectional correspondence: for any separable win-martingale in our class, there exists a concave pool-value function that achieves uniform LVR, and conversely, any sufficiently regular pool-value function induces a win-martingale under which the corresponding AMM is uniform, linking AMM geometry with the underlying dynamics of a price process.

% On the AMM side, we represent an invariant through its \emph{pool-value function} $V(p)$:
% the marked-to-market value of the arbitrage-free AMM when the fair YES price is $p$. For a
% separable win-martingale, It\^o's formula implies that instantaneous LVR is
% \[
%     \LVR_t
%     =
%     -\frac{1}{2}V''(P_t)\frac{G(P_t)^2}{h(t)^2}.
% \]
% Thus the state dependence of LVR is governed by the product of the AMM's curvature
% $V''$ and the price process's local variance $G^2$. Requiring LVR to be proportional to
% pool value across all price states yields the boundary value problem
% \[
%     G(p)^2 V''(p)+\beta\cdot V(p)=0,
%     \qquad
%     V(0)=V(1)=0,
%     \qquad
%     V(p)>0 \quad \text{for } p\in(0,1).
% \]
% This equation is the central object of the paper. It links AMM geometry to belief dynamics:
% the volatility profile determines the curvature needed for uniform losses, and conversely,
% the curvature of an AMM value function determines the belief process under which it is
% uniform.

\paraheader{Contributions.}
Our work makes four contributions at the intersection of mechanism design for subsidized price discovery, stochastic process theory, and prediction market AMMs:
\begin{enumerate}[(i)]
    \item We provide a general framework for designing AMMs with uniform LVR across prices in binary prediction markets. We define uniformity as the requirement that LVR, as a fraction of pool value, be independent of the price after normalizing by informational time, and show that this condition is equivalent to a second-order BVP.
    \item We establish a bidirectional correspondence between win-martingale price processes and AMM invariants. Starting from a separable win-martingale, we show that a corresponding concave pool-value function can be constructed by solving the uniformity BVP. Conversely, starting from a sufficiently regular pool-value function, we construct a volatility profile under which the AMM has uniform LVR. Thus, a price process determines the AMM geometry needed to equalize losses, and an AMM invariant can be interpreted as embedding a process under which prices evolve.
    \item We provide a set of examples that make this correspondence explicit. For canonical
      win-martingales with natural applications, we derive the associated uniform AMMs where
      possible.\footnote{For example, the Wright-Fisher / Jacobi volatility profile leads to a
        constant-elasticity invariant while the logistic / log-odds diffusion corresponds to the constant product market maker.} We also take familiar AMMs, such as the constant product market maker and the logarithmic market scoring rule, and identify win-martingale dynamics under which they are uniform.
    \item We extend our framework to a dynamic setting. While (static) uniformity controls how LVR varies across price states, liquidity providers and protocol designers also care about \textit{when} losses are incurred. Once a uniform AMM value function has been fixed, we show that a deterministic liquidity schedule can implement a prescribed target expected cumulative loss schedule. The dynamic setting separates and solves two design problems: the invariant controls the statewise distribution of the subsidy, while the liquidity schedule controls its time profile.
\end{enumerate}

\paraheader{Related literature.}
This paper contributes to three strands of literature. The first is prediction market mechanism design. Market scoring rules, introduced by \citet{Hanson03,hanson2007lmsr}, provide the foundational mechanism for automated price discovery in prediction markets, with the LMSR as the canonical example. \citet{chen2007utility} develop a utility-theoretic framework for bounded-loss market makers, generalizing Hanson's construction. These mechanisms trade contingent claims whose prices can be interpreted as probabilities \citep{arrow_debreu_1954,wolfers2004prediction,wolfers2006interpreting}. Empirical work has documented the forecasting performance of prediction markets \citep{berg2008prediction}, while more recent work studies price formation and information incorporation in field prediction markets \citep{bossaerts2024price}. Closely related to our focus on stochastic models of probability dynamics, \citet{archak2009modeling} study volatility in prediction markets, and \citet{dalen2025blackscholes} proposes a logit jump-diffusion framework for prediction-market belief risk.

The second strand is automated market making in decentralized
finance. \citet{frongillo2024axiomatic} establish the equivalence between cost-function
prediction-market mechanisms and constant-function market makers. Subsequent work studies
invariant geometry, liquidity provision, and concentrated-liquidity designs
\citep{uniswapv3whitepaper2021,angeris2024geometry,capponi2025liquidity}. \citet{milionis2022automated}
formalize loss-versus-rebalancing as the adverse-selection cost of AMM liquidity.
%In an earlier unpublished version of this paper,
\citet{paradigm2024pmamm} apply this idea to prediction markets
and derive a uniform-LVR AMM for Gaussian score dynamics.
We generalize this construction by
characterizing uniform AMMs for broad classes of win-martingale price processes and by studying
the converse mapping from AMM invariants to compatible belief dynamics.\todo{mention this is an
  earlier version in non-anonymous version}

The third strand is martingale theory and belief dynamics. Martingales are fundamental to fair pricing under information filtrations \citep{doob1953stochastic}. In binary prediction markets, the relevant objects are bounded martingales that terminate at zero or one. Recent work on win-martingales studies entropy, activity, and optimal-transport properties of such processes \citep{backhoff2023mostexciting,backhoff2024exciting,backhoff2024wasserstein}. Our paper uses win-martingales as models of prediction-market prices and connects their volatility profiles to implementable AMM invariants.

%%% Local Variables:
%%% mode: LaTeX
%%% TeX-master: "main"
%%% End:

\section{Prediction Market Price Processes}\label{sec:processes}

% \todo[inline]{
% Setup:

% token $X$ resolves to zero or one depending on if an event happens by by time $T$

% definition: win-martingale

% want price process $P_t$ in $[0,1]$ for $t < T$

% want $P_T\in \{0,1\}$

% should be a martingale (we will work under risk neutral measure, zero interest rates)
% }

We model the evolution of fair prices in a binary prediction market. Consider a
contract that resolves at a fixed terminal time $T>0$ to one of two mutually exclusive outcomes. Let $X$ denote the token that pays one unit of numeraire if the event occurs and zero otherwise, and $Y$ denote the complementary token; we refer to the $X$ and $Y$ tokens as YES and NO tokens, respectively. Since exactly one of the two outcomes occurs, the terminal payoff from holding both tokens is one, so the fair prices of the two tokens must sum to one at all times. We henceforth focus on the YES-token's price without loss of generality.

Let $(\Omega,\mathcal{F},\{\mathcal{F}_t\}_{0\leq t\leq T},\mathbb{P})$ be a filtered probability space representing the information available to market participants over time. We denote the process for the fair price of token X (or the YES-token) by $\{P_t\}$. Assuming a risk-neutral environment with zero interest rates, the fair price process must be a martingale. Thus, for $s \leq t\leq T$, we have
% \[
$
    \mathbb{E}[P_t \mid \mathcal{F}_s] = P_s.
$
%\]
This condition captures the usual no-arbitrage interpretation of prediction market prices: conditional on the information currently available, the expected future probability of the event is equal to the current probability. Prediction market prices differ from ordinary asset prices because the terminal value is binary. At resolution, the uncertainty is fully revealed and the price must collapse to one of the two endpoints, zero or one. This motivates the following class of processes.

\begin{definition}[Win-martingale]
A process $\{P_t\}_{0\leq t\leq T}$ is a \underline{win-martingale} if it is a martingale taking values in $[0,1]$ and satisfies $P_T \in \{0,1\}$ almost surely.
\end{definition}

Win-martingales provide a natural reduced-form model of belief dynamics in prediction markets. The
martingale property encodes temporal consistency, while the terminal condition encodes resolution of the event. Throughout the paper, we focus on continuous price processes, so that beliefs evolve through the gradual arrival of information rather than through predictable jumps. In particular, we model $P$ as an It\^o diffusion without drift,
\[
    dP_t = \sigma(P_t,t)\,dW_t,
\]
where $W$ is a Brownian motion adapted to $\{\F_t\}$ and $\sigma$ is a state- and time-dependent volatility function. The absence of a drift term is forced by the martingale property, while the volatility function determines how strongly prices react to new information at different price levels and times before resolution.

\subsection{Separable Win-Martingales}

% \todo[inline]{
% (could be at end of prior section)

% want to work with cts price processes, hence diffusions

% m.g. property forces  $dP_t =  \sigma(P_t,t)\,dW_t$

% we will make a more restrictive assumption, separable across time and price

% all diffusive models that have been proposed in the literature satisfy this property
% }

To obtain a tractable yet flexible class of price processes, we impose a separable structure on the diffusion coefficient.
\begin{definition}[Separable win-martingale]
A win-martingale $\{P_t\}$ is \underline{separable} if it satisfies
\begin{equation}\label{eq:sde}
  dP_t = \frac{G(P_t)}{h(t)}\,dW_t
\end{equation}
for some $G:[0,1]\to\mathbb{R}_+$ and $h:[0,T)\to\mathbb{R}_+$.
\end{definition}
where $G$ determines the state dependence of volatility and $h$ determines the rate at which information  is released.

This specification separates two conceptually distinct features of belief dynamics. The function $G$ describes how sensitive the market price is to new information at a given price level. For example, a model with large $G(p)$ near $p=1/2$ and small $G(p)$ near the endpoints says that beliefs are most responsive when the outcome is uncertain and become increasingly sticky as the event becomes almost resolved. The function $h$, by contrast, controls the temporal pattern of information arrival. Smaller values of $h(t)$ correspond to faster information arrival and hence greater instantaneous volatility in calendar time. In this sense, $h(t)$ provides a flexible reduced-form way to encode different informational environments, such as steady information flow, late-breaking revelation, or accelerating news arrival near the event date.

Formally, this interpretation can be made precise via a time change. Define the \emph{informational time}
\[
    \tau(t) := \int_0^t \frac{1}{h(s)^2}\,ds,
\]
which measures the cumulative information released by calendar time $t$. The price process can then be written as a time change of the time-homogeneous diffusion
\[
    dQ_\tau = G(Q_\tau)\,dB_\tau, \qquad P_t = Q_{\tau(t)},
\]
where $B$ is a Brownian motion in informational time. In this representation, $G$ governs the diffusion under its natural information clock, while $h$ rescales calendar time to informational time. We impose the following regularity and boundary conditions.

\begin{assumption}
\label{ass:win-martingale}
The functions $G$ and $h$ satisfy:
\begin{enumerate}[(i)]
    \item $G$ is continuous on $[0,1]$, with $G(0)=G(1)=0$ and $G(p)>0$ for all
    $p\in(0,1)$;
    \item $G$ is locally Lipschitz on $(0,1)$;
    \item $\displaystyle\liminf_{p\searrow0}\frac{G(p)}{p}>0$ and $\displaystyle\liminf_{p\nearrow1}\frac{G(p)}{1-p}>0$;
    \item $\tau(t) = \displaystyle \int_0^t \frac{1}{h(s)^2}\,ds < \infty$ for all $t\in[0,T)$ and
      $\tau(T) = \displaystyle\int_0^T \frac{1}{h(s)^2}\,ds = \infty$.
\end{enumerate}
\end{assumption}

Conditions (i) and (ii) ensure that the interior diffusion is well posed: the endpoints are absorbing, volatility is strictly positive in the interior, and the stochastic differential equation has a unique solution up to the boundary. Condition (iii) imposes endpoint regularity; it is not needed for the process to be a win-martingale, but it will be used later to obtain regularity of the uniform AMM boundary value problem. Condition (iv) ensures that only finite informational time elapses over any strict subinterval of \([0,T)\), while infinite informational time elapses by the resolution date $T$, guaranteeing resolution of the process at time $T$. Since $P$ is a bounded local martingale, it is a true martingale so $P$ is a win-martingale; \cref{prop:separable-win-martingale} formalizes this result.

\begin{proposition}
\label{prop:separable-win-martingale}
Under Assumption~\ref{ass:win-martingale}, if $P$ satisfies \eqref{eq:sde} with
$P_0=p_0\in(0,1)$, then:
\begin{enumerate}[(i)]
    \item $(P_t)_{0\leq t\leq T}$ is a win-martingale.
    \item For every $t\in[0,T]$, the conditional variance of the terminal outcome satisfies
    \[
      \Var[P_T\mid\mathcal{F}_t]
      = P_t(1-P_t)
        =
        \mathbb{E}\!\left[
            \int_t^T \frac{G(P_s)^2}{h(s)^2}\,ds
            \,\middle|\,\mathcal{F}_t
        \right].
    \]
\end{enumerate}
\end{proposition}

% This separable class includes many natural models of prediction-market belief dynamics. The profile $G$ captures how volatility varies across probability states, while $h$ captures the market's information-release schedule. This separation will be crucial in the sequel: the static shape of a uniform AMM is determined by $G$, whereas the time profile of expected losses is governed by $h$ and by the liquidity schedule chosen by the market maker.

Part (ii) gives a calibration check: every admissible choice of \(G\) and \(h\) must allocate the remaining Bernoulli variance \(P_t(1-P_t)\) across future informational time.

% \todo[inline]{
% idea: start with a bounded m.g. on $[0,\infty)$, this must have a limiting distribution, we will require that limiting distribution be concentrated on $\{0,1\}$, this leads to conditions on $G$, rescale time so this convergence happens on $[0,T]$, this leads to conditions on $h$

% write out time scaling mapping $[0,+\infty] \rightarrow [0,T]$

% Proposition: if we make these assumptions, resulting price process is a win-martingale

% show variance / quadratic variation decomposition
% }

\subsection{Examples of Volatility Profiles}\label{sec:wm-example}

We now discuss several canonical choices of the volatility profile $G$ in the win-martingale dynamics, each of which encodes a different model of how beliefs evolve in a binary-outcome market. The function $G$ determines how sensitive prices are to incoming information at different probability levels, while $h$ governs the time profile of information release. Unless an information clock is displayed explicitly, these profiles should be combined with any clock \(h\) satisfying Assumption~\ref{ass:win-martingale}(iv). These examples illustrate the range of belief dynamics that can arise in prediction markets, along with their economic interpretation and potential applications.

\vspace{1em}

\paraheader{Absorbed Brownian motion.}
This is the simplest benchmark process, corresponding formally to Brownian motion on $(0,1)$ with absorption at the endpoints, is
\[
G(p)=\mathbbm{1}_{\{0<p<1\}}.
\]
The instantaneous volatility of beliefs is constant across interior price states: information moves the market by the same magnitude regardless of the contract's current trading price.

\vspace{1em}

\paraheader{Wright-Fisher / Jacobi diffusion.}
The volatility profile
\[
G(p)=\sqrt{p(1-p)}
\]
is the classical Wright-Fisher or Jacobi diffusion coefficient \citep{ethier1986markov}. This profile, with natural applications in population genetics, results from the continuous-time version of an allele frequency process over successive generations, see \cref{app:supp-price-processes} for the diffusion-limit derivation.

\vspace{1em}

\paraheader{Logistic / log-odds diffusion.}
The volatility profile
\[
    G(p)=p(1-p)
\]
corresponds to a price process whose randomness is additive in log-odds space. To see this, define the log-odds process
\[
    \mathcal{O}_t := \log\frac{P_t}{1-P_t}.
\]
Under this profile, the log-odds process has constant instantaneous volatility up to the information clock \(h(t)\). This specification is natural when evidence is approximately additive in log-likelihood-ratio space, as in repeated Bayesian updating from conditionally independent signals; see \cref{app:supp-price-processes}.

% \vspace{1em}

% \paraheader{General Beta-type family.}
% The parametric family of profiles given by
% \[
% G(p)=p^a(1-p)^b
% \]
% for $a,b>0$ generalizes the previous cases and provides a flexible way to control boundary behavior asymmetrically. The parameters $a$ and $b$ determine how quickly volatility vanishes near $0$ and $1$, respectively. When $a=b=\tfrac12$, one recovers the Wright-Fisher diffusion; when $a=b=1$, one obtains the logistic diffusion. This family allows one to encode asymmetric informational environments: for example, if bad-news realizations are typically sharper or more decisive than good-news realizations, one may choose $a\geq b$ to capture the differential volatility decay near the two endpoints. 

% This can be useful for contracts where one side of the event is easier to verify or more exposed to sudden revelation, such as default/no-default markets, litigation outcomes, regulatory approvals, or prediction contracts with strongly asymmetric upside and downside information structures. From the AMM-design perspective, this family offers a rich testbed for studying how invariant curvature should respond to different kinds of boundary asymmetry.

\vspace{1em}

\paraheader{Gaussian score dynamics.}
Gaussian score dynamics \citep{archak2009modeling,paradigm2024pmamm} can arise when the event outcome is determined by the sign of an underlying continuous fundamental process $\{Z_t\}$, such as a latent score differential, with common examples including the difference in score between two teams in a sporting event or in votes between two candidates in a political election. If $\{Z_t\}$ evolves as an arithmetic Brownian motion with zero drift and volatility $\sigma$, and the YES token pays off when the terminal fundamental is nonnegative, the induced price process has
\[
    G(p)=\phi(\Phi^{-1}(p)),
    \qquad
    h(t)=\sqrt{T-t},
\]
regardless of $\sigma$. Here \(\Phi\) and \(\phi\) denote the standard normal distribution function and density, respectively. A short derivation is given in \cref{app:supp-price-processes}.

% This model is a natural reduced form for markets in which the binary outcome is induced by an underlying continuous quantity crossing a threshold. In sports markets, $Z_t$ can be interpreted as a latent score or performance differential; in election markets, it can represent a latent vote-share margin; and in macroeconomic or policy markets, it can represent a continuous fundamental whose sign determines the event. The Gaussian specification is especially tractable because posterior probabilities are obtained by passing the normalized latent state through the probit link, and the resulting price volatility is highest near $p=1/2$ and vanishes near the endpoints.

\vspace{1em}

\paraheader{The most exciting game.} \citet{backhoff2023mostexciting} formally show that the win-martingale that maximizes an information-theoretic notion of uncertainty among all win-martingales as defined earlier in this section is 
\[
    dP_t
    =
    \frac{\sin(\pi P_t)}{\pi\sqrt{T-t}}\,dW_t.
\]
In the notation of our separable win-martingale specification, we have
\[
    G(p)=\frac{\sin(\pi p)}{\pi}, \qquad h(t)=\sqrt{T-t}.
\]
\citet{backhoff2023mostexciting} refer to this process as ``the most exciting game.''

\subsection{Examples of Information Release Rates}

We now discuss several choices for the information clock $h$. Whereas the volatility profile $G$ determines how sensitive beliefs are to information at different price levels, $h$ shapes the accumulation of information and therefore controls whether price discovery is smooth or concentrated near resolution, for example. Various choices of $h$ can induce different term structures of volatility and adverse selection while leaving the state dependence of belief updates unchanged. 

\vspace{1em}

\paraheader{Square-root clock.}
The schedule
\[
h(t)=\sqrt{T-t}
\]
releases information at a rate proportional to the inverse square root of remaining time to
resolution. This can be thought of as a canonical clock as it appears as the schedule for Gaussian
score dynamics and the most exciting game in \citet{backhoff2023mostexciting}. Another useful
property of this clock is that the remaining quadratic variation of the process is $(T-t)$, which
is linear in the time to resolution.\todo{doesn't make sense to me, QV should also depend on $G$ --Ciamac}

\vspace{1em}

\paraheader{Brody-Hughston-Macrina (BHM) clock.}
The Brody-Hughston-Macrina clock is given by
\[
    h(t)=\frac{T-t}{\sigma T},
\]
for $\sigma>0$. This clock arises naturally in the information-based asset-pricing framework of
\citet{brody2008information}. For a binary terminal payoff, the induced posterior process satisfies a separable diffusion of the logistic form
\[
    dP_t
    =
    \frac{\sigma T}{T-t} \cdot P_t(1-P_t)\,dW_t.
\]
Equivalently, in our notation, we ahve
\[
    G(p)=p(1-p),
    \qquad
    h(t)=\frac{T-t}{\sigma T}.
\]
The BHM clock therefore gives a structural microfoundation for the logistic/log-odds diffusion described above; details are in \cref{app:supp-price-processes}.

\vspace{1em}

\paraheader{General power-law family.}
A convenient general family is
\[
h(t)=c(T-t)^\alpha,
\]
where $c>0$ controls the overall speed of information arrival and $\alpha\geq1/2$ governs its temporal shape. This family provides a parsimonious way to interpolate across different informational environments. For example, the square root clock is encoded by $c=1$ and $\alpha=1/2$, and the Brody-Hughston-Macrina clock is encoded by $c=1/(\sigma T)$ and $\alpha=1$.
The lower bound on \(\alpha\) is sharp for terminal resolution: since
\[
  \tau(t) = \int_0^t h(s)^{-2}\,ds
    =
    c^{-2}\int_0^t (T-s)^{-2\alpha}\,ds,
\]
the accumulated informational time diverges as \(t\uparrow T\) if and only if \(2\alpha\ge1\).

%%% Local Variables:
%%% mode: LaTeX
%%% TeX-master: "main"
%%% End:

\section{Automated Market Makers}\label{sec:amms}

We now introduce an automated market maker (AMM) model used throughout the paper. We consider a binary prediction market with two outcome tokens, denoted by $X$ and $Y$, as detailed in \cref{sec:processes}. Notably, rather than exchanging outcome tokens against a numéraire, the AMM exchanges the two outcome tokens against each other. A trader who wants to buy exposure to
the event trades into the $X$ token and out of the $Y$ token, while a trader who wants to sell exposure does the reverse.

Liquidity provision and trading are governed by a reserve-based invariant. Let $x$ and $y$ denote the AMM's reserves of the two outcome tokens. We represent the AMM by a pricing function $F:\mathbb{R}_+^2\times \mathbb{R}_+ \to \mathbb{R}$ where the first two arguments are the token reserves and the third argument $L>0$ is a liquidity scale or pool-depth parameter. The feasible
reserve states for liquidity level $L$ are the points satisfying
\[
    F(x,y;L)=0.
\]
A trade moving reserves from $(x,y)$ to $(x+\Delta_X,y+\Delta_Y)$ is feasible if it preserves the invariant
\[
    F(x+\Delta_X,y+\Delta_Y;L)=F(x,y;L)=0.
\]
The signs of $\Delta_X$ and $\Delta_Y$ depend on the direction of the trade. For example, if the trader buys $X$ and pays in $Y$, then the pool's $X$ reserve decreases and its $Y$ reserve increases.

Given an external fair price $p$ for token $X$, arbitrageurs move the AMM to the point on the invariant curve that is least valuable at prices $(p,1-p)$, thus extracting as much value from the pool as possible. This gives the pool-value problem \citep{milionis2022automated}
\[
    V_L(p)
    =
    \inf_{(x,y)\in\mathbb{R}_+^2}
    \left\{
        px+(1-p)y:
        F(x,y;L)=0
    \right\}.
\]
When the liquidity scale is fixed, we write $V(p)$ for $V_L(p)$. The function $V$ is the
marked-to-market value of the arbitraged AMM position as a function of the fair probability
$p$. This representation is convenient because loss-versus-rebalancing and the uniformity
condition depend on the curvature of $V$. We impose the following regularity conditions on the invariant.

\begin{assumption}
\label{ass:pricing-function}
$F$ is twice continuously differentiable on $\mathbb{R}_{++}^2$ and
satisfies:
\begin{enumerate}[(i)]
    \item $F_x,F_y>0$ and $F_y^2F_{xx}-2F_xF_yF_{xy}+F_x^2F_{yy}<0$;
    \item $\lim_{x\downarrow 0}\frac{F_x}{F_y}=\infty$, $\lim_{x\uparrow \infty}\frac{F_x}{F_y}=0$, $\lim_{y\downarrow 0}\frac{F_x}{F_y}=0$ and $\lim_{y\uparrow \infty}\frac{F_x}{F_y}=\infty$;
    \item $\inf_{(x,y):F(x,y,L)=0}y=0$ and $\inf_{(x,y):F(x,y,L)=0}x=0$.
\end{enumerate}
\end{assumption}

Condition (i) gives monotonicity and the curvature needed for a well-defined marginal price. Conditions (ii) and (iii) are Inada-type boundary conditions ensuring all marginal prices \(p\in(0,1)\) are supported. Under these conditions, \cref{prop:pool-value-well-posed-main} gives a unique arbitraged pool value.

\begin{proposition}[Well-posed pool value problem]
\label{prop:pool-value-well-posed-main}
Suppose Assumption~\ref{ass:pricing-function} holds. Given a fixed liquidity level $L>0$, for every $p\in(0,1)$, the pool-value problem has a unique solution. Moreover, $V_L$ is differentiable and concave on $(0,1)$, and extends continuously to $[0,1]$ with $V_L(0)=V_L(1)=0$.
\end{proposition}

A direct consequence is that the pool's terminal value vanishes: since \(P_T\in\{0,1\}\), we have \(V_L(P_T)=0\) almost surely.

\cref{prop:demand-functions} records how to recover demand functions and an invariant representation from the pool-value function.

\begin{proposition}[Demand and invariant functions from the pool-value function]
\label{prop:demand-functions}
Suppose Assumption~\ref{ass:pricing-function} holds, and let $(X(p),Y(p))$ be the unique solution to the pool-value problem at $p\in(0,1)$. Then:
\begin{enumerate}[(i)]
    \item The demand functions are given by
    \[
        X(p)=V(p)+(1-p)\cdot V'(p),
        \qquad
        Y(p)=V(p)-p\cdot V'(p).
    \]
    \item A reconstructed invariant with the same pool-value function is given by
    \[
    F(x,y,L)=\inf_{p\in(0,1)} \frac{px+(1-p)y}{V(p)}-L.
    \]
\end{enumerate}

\end{proposition}

\subsection{Examples of Automated Market Makers}\label{sec:amm-example}

To connect our framework to familiar market designs, we briefly review two canonical AMMs that have played a central role in both decentralized finance and prediction markets. Each can be represented equivalently by an invariant, a value function, and reserve demand functions.

\vspace{1em}

\paraheader{Constant product market maker (CPMM).}
The CPMM is defined by the invariant
\[
F(x,y)=xy=L ^2,
\]
where \(x\) and \(y\) denote the reserves of the two outcome tokens and \(L >0\) is a liquidity scale parameter. Solving the pool-value problem yields pool-value and demand functions
\[
V(p)=2L \sqrt{p(1-p)},
\qquad
X(p)=L \sqrt{\frac{1-p}{p}},
\qquad
Y(p)=L \sqrt{\frac{p}{1-p}}.
\]
Note that value here is concentrated near the center of the price space and vanishes at the endpoints, while being linear in the amount of liquidity $L$ deployed to the AMM.
% Thus the constant-product AMM concentrates value near the center of the state space and thins out toward the boundaries, reflecting the fact that liquidity is deepest when the market is most uncertain. Differentiating gives
% \[
% V''(p)=-\frac{L }{2\,[p(1-p)]^{3/2}},
% \]
% so its curvature increases sharply near \(0\) and \(1\), implying high price impact and large local adverse-selection sensitivity in those regions. In our framework, this invariant corresponds formally to the volatility profile
% \[
% G(p)=p(1-p),
% \]
% that is, the logistic or log-odds diffusion. Accordingly, the constant-product AMM may be interpreted as the canonical invariant matched to a market in which information arrives additively in log-odds space.

\vspace{1em}

\paraheader{Logarithmic market scoring rule (LMSR).}
The LMSR has an invariant representation of
\[
F(x,y;L)=1-e^{-x/L }-e^{-y/L }=0,
\]
where \(L >0\) controls the scale of liquidity. This is the invariant-based analogue of Hanson’s cost-function market maker \citep{hanson2007lmsr} and is an important benchmark in prediction-market design. The corresponding value and demand functions are
\[
V(p)
=
L  H(p),
\qquad
X(p)=-L \log p,
\qquad
Y(p)=-L \log(1-p),
\]
where $H(p)=-(p\log p+(1-p)\log(1-p))$ is the entropy function. Like the constant product market
maker, pool value is highest for intermediate price ranges, vanishes towards the endpoints, and is
linear in the liquidity scale. Observing that the entropy function is maximized at
$p=\sfrac{1}{2}$, we have that $V(p) \leq L H(\sfrac{1}{2}) = L \log 2$, which recovers the bound
of \citet{hanson2007lmsr} on the total pool loss.

% Unlike constant product, LMSR maintains strictly positive reserves throughout the interior and exhibits a more gradual decay of liquidity near the boundaries. Since
% \[
% V''(p)=-L \left(\frac{1}{p}+\frac{1}{1-p}\right)
% =
% -\frac{L }{p(1-p)},
% \]
% its curvature still diverges at the endpoints, but less sharply than in the constant-product case. In our framework, the associated volatility profile is
% \[
% G(p)=\sqrt{p(1-p)\,H(p)},
% \qquad
% H(p):=-p\log p-(1-p)\log(1-p),
% \]
% so LMSR is naturally matched to a process in which local volatility is proportional not only to posterior variance \(p(1-p)\), but also to the square root of entropy. This highlights the close connection between LMSR and information-theoretic notions of uncertainty.

\subsection{Loss-Versus-Rebalancing}

We next quantify the cost borne by liquidity providers when the AMM is arbitraged back to the current fair price. Because AMM quotes are functions of pool reserves, they update only when trades occur. If the external fair price moves before the pool is traded against, then informed traders can profit by moving the AMM from its stale reserve allocation to the allocation consistent with the new price. The resulting transfer from the pool to arbitrageurs is called loss-versus-rebalancing (LVR).

In our prediction-market setting, this cost also has a natural positive interpretation. The fair price process $P_t$ represents the information available to traders who observe the event-relevant signal. When the AMM price is stale, arbitrageurs who trade against the pool move its quote toward the observable fair price. Thus LVR can be viewed either as an adverse-selection loss to LPs or as the subsidy paid by the AMM to informed traders for correcting prices.

Consider a separable win-martingale with dynamics
\[
    dP_t = \frac{G(P_t)}{h(t)}\,dW_t .
\]
and an AMM with pool-value function $V$. With slight abuse of notation, denote the arbitraged mark-to-market value of the pool at time $t$ as
\[
    V_t := V(P_t).
\]
Applying Itō's formula to $V(P_t)$ decomposes the pool value dynamics into a martingale component and a predictable drift; the drift is nonpositive by concavity of $V$, and LVR is its absolute value.

\begin{proposition}[Loss-versus-rebalancing decomposition]
\label{prop:lvr-decomposition}
Let $\{P_t\}$ be a separable win-martingale with volatility profile $G$ and information release rate $h$, and let $V$ be a twice-differentiable, concave pool-value function of an AMM. Define the \underline{instantaneous LVR rate}
\[
    \LVR_t \coloneqq -\frac{1}{2}\cdot V''(P_t)\cdot\frac{G(P_t)^2}{h(t)^2}\,\geq\,0.
\]
Then:
\begin{enumerate}[(i)]
    \item The pool value $V_t = V(P_t)$ satisfies
    \[
        dV_t
        =
        V'(P_t)\,\frac{G(P_t)}{h(t)}\,dW_t
        - \LVR_t\,dt.
    \]
    \item\label{en:v-lvr} The expected loss of the pool by time $t$ is given by
    \[
        V(P_0) - \E\!\left[V(P_t)\right]
        =
        \E\!\left[\int_0^t \LVR_s\,ds\right].
    \]
\end{enumerate}
\end{proposition}

The drift term \(-\LVR_t\) is the predictable loss from continuous arbitrage to the fair
price. Since \(P_T\in\{0,1\}\) and \(V(0)=V(1)=0\), overall terminal expected loss is
\[
  V(p_0)=\E\left[\int_0^T \LVR_t\,dt\right].
\]

%%% Local Variables:
%%% mode: LaTeX
%%% TeX-master: "main"
%%% End:

\section{Static PM-AMMs: Uniform LVR Over Price States}\label{sec:static}

We now introduce the notion of uniform loss-versus-rebalancing. Given a separable win-martingale and an AMM with pool-value function $V$, the
instantaneous LVR rate as a fraction of pool value is
\[
    \frac{\LVR_t}{V(P_t)}
    =
    -\frac{1}{2}
    \cdot \frac{V''(P_t)G(P_t)^2}{V(P_t)}
    \cdot
    \frac{1}{h(t)^2}.
\]
For a fixed time $t$, the factor $h(t)^{-2}$ is common across all price states, so the state dependence of relative LVR is entirely captured by
$V$ and $G$. A (static) uniform AMM is one for which this quantity is constant for all $p\in(0,1)$; normalizing by $h(t)^{-2}$ and rearranging yields
\[
    -\frac{1}{2}V''(p)G(p)^2 = \alpha  \cdot V(p).
\]
Uniformity therefore concerns the cross-sectional dependence of LVR on price states, after normalizing by the common information-release clock. Setting $\beta:=2\alpha$ yields the boundary value problem (BVP) for the pool-value function.

\begin{definition}[Uniformity BVP]
\label{def:static-uniform-amm}
For a separable win-martingale with volatility profile $G$, the \underline{uniformity boundary value problem} is the eigenpair problem in $(\beta,V)$ such that
\[
    G(p)^2 \cdot V''(p)+\beta \cdot V(p)=0,
    \qquad \text{for all } p\in(0,1),
\]
with endpoint and interior conditions
\[
    V(0)=V(1)=0,
    \qquad
    V(p)>0 \quad \text{for all } p\in(0,1).
\]
\end{definition}

The uniformity BVP has a direct economic interpretation. The term
$G(p)^2$ measures the local quadratic variation of the fair price per unit of informational time, while $V''(p)$ measures the curvature of the AMM. Uniformity requires these two effects to offset each other so that the product $-V''(p)\cdot G(p)^2$ is proportional to the pool value $V(p)$ itself. Where prices fluctuate more intensely, the value function must have less curvature per dollar of pool value; where prices fluctuate less intensely, the invariant can sustain more curvature without producing excess relative LVR. We can think of uniformity as a matching condition between the stochastic dynamics of information arrival and the geometry of the AMM invariant.

This motivates the following terminology. For a given fair-price process that is a separable win-martingale, an AMM is a \textit{uniform AMM} or \textit{prediction market AMM (PM-AMM)} if its pool-value function satisfies the uniform-LVR condition for that process. In other words, a PM-AMM is an AMM whose adverse-selection losses are calibrated to the price dynamics of the underlying event so that instantaneous LVR is proportional to pool value at every price state. This emphasizes that PM-AMMs are not necessarily generic liquidity mechanisms, but rather are tailored to the specific dynamics of an underlying event. % In this sense, a PM-AMM is the invariant-based analogue of a proper prediction-market scoring rule: it is a market-making rule explicitly adapted to the dynamics of probabilistic information aggregation.

We now discuss several normative properties that distinguish uniform LVR from alternative loss distributions and motivate it as a design criterion for prediction market AMMs.

\paraheader{Closed-form dynamic liquidity management.}
Under any AMM, from \Cref{prop:lvr-decomposition}(\ref{en:v-lvr}), the expected pool value $m(t)=\E[V(P_t)]$ evolves as
\[
  m'(t)=\frac{1}{2h(t)^2}\E\!\left[V''(P_t)\,G(P_t)^2\right].
\]
Under a uniform AMM, this reduces to the closed-form ODE $m'(t)=-\frac{\beta}{2h(t)^2}\,m(t)$,
which depends only on the eigenvalue $\beta$ and the information clock $h$, not on the
distribution of $P_t$. As a consequence, the liquidity schedule that implements any prescribed
target expected loss schedule has an explicit closed-form expression involving only $\beta$, $h$,
$V(p_0)$, and the target schedule $D'(t)$; see \cref{sec:dynamic}. For a non-uniform AMM, the
expectation $\E[V''(P_t)\,G(P_t)^2]$ depends on the full distribution of $P_t$, so computing the
corresponding liquidity schedule requires solving the Kolmogorov forward equation for the
transition density of the price process, and is no longer available in closed form.

\paraheader{Robustness to liquidity provider behavior.}
If liquidity providers can dynamically enter and exit the pool, each LP faces an instantaneous loss rate per dollar deployed equal to $\LVR_t/V(P_t)$. If this ratio varies with the current price $P_t$, sophisticated LPs have an incentive to withdraw capital when the price is in a high-relative-LVR state and deposit when it is in a low-relative-LVR state. This creates adverse selection among LPs themselves: the mechanism becomes under-capitalized precisely in the states where the subsidy cost is highest, undermining the goal of uniform price discovery. Under a uniform AMM, every LP faces the same relative loss rate regardless of the current price. The participation decision reduces to comparing this uniform rate against the LP's opportunity cost of capital, independently of the price state. This eliminates the incentive to time entry and exit based on the current belief, making the mechanism robust to LP sophistication.

\paraheader{Minimax fairness across price states.}
Among all admissible pool-value functions, the uniform AMM minimizes the worst-case relative LVR across price states.

\begin{theorem}[Minimax optimality of uniform LVR]\label{prop:minimax}
Let $G$ satisfy \cref{ass:win-martingale} and let $(\beta,V^*)$ solve the uniformity BVP (\cref{def:static-uniform-amm}). For any $V\in C([0,1])\cap C^2((0,1))$ with $V(0)=V(1)=0$ and $V(p)>0$ for $p\in(0,1)$,
\[
    \sup_{p\in(0,1)}\frac{-V''(p)\,G(p)^2}{V(p)}
    \;\geq\;
    \beta,
\]
with equality if and only if $V$ is a positive multiple of $V^*$.
\end{theorem}

\noindent Thus a uniform AMM equalizes the worst-case relative LVR across price states: no price state bears a disproportionate share of the subsidy cost relative to pool value.

\paraheader{Budget legibility.}
Under uniform LVR, the subsidy functions as an effective interest rate on deployed capital: the
subsidizer loses a fixed fraction of pool value per unit of informational time, regardless of
where prices are. Budget accounting is therefore path-independent --- expected cumulative losses depend only on the liquidity schedule and the information clock, not on the realized price trajectory. This allows the subsidizer to make capital allocation decisions without forecasting the future path of beliefs.

\paraheader{Scope and alternatives.}
We emphasize that uniformity is one design choice within a broader family. A mechanism designer
who values information more at certain prices --- for example, near $p=1/2$ where binary decisions
are most sensitive to the probability --- might prefer a non-uniform design that concentrates subsidy in those regions. The paper's framework extends naturally to such designs: the uniformity BVP can be generalized by replacing the constant $\beta$ with a price-dependent weight function, yielding a modified eigenvalue problem. Uniform LVR provides a natural baseline as the design that treats all price states equally relative to pool value, and delivers the cleanest analytical framework via the properties above.

\subsection{From Win-Martingale to AMM Invariant}

We begin illustrating the correspondence between win martingales and uniform-loss invariants by showing that pool-value functions with uniform LVR can be constructed for the class of win-martingales we consider.

\begin{theorem}[Constructing a uniform AMM for a win-martingale]\label{prop:WMtoPV}
For a separable win-martingale satisfying \cref{ass:win-martingale} there exists a pair $(\beta,V)$ where $\beta>0$, and $V:[0,1]\to\R_+$ is continuous and twice-differentiable (in the interior) function such that $(\beta,L\cdot V)$ solves the uniformity BVP under $G$ for any $L>0$. Moreover, $V$ is concave and unimodal on $[0,1]$.
% \begin{gather*}
%     \beta\cdot V(p)+G(p)^2\cdot V''(p) = 0, \\[0.25\baselineskip]
%     V(0) = V(1) = 0.
% \end{gather*}
\end{theorem}

\cref{prop:WMtoPV} shows that win-martingales admit concave pool-value functions \(V\) whose instantaneous LVR is proportional to pool value across price states. Given this \(V\), one can scale by the liquidity parameter \(L\) and recover demand functions or an invariant representation.

The key concept is that the uniformity BVP can be written as a singular Sturm-Liouville problem.\footnote{The problem is singular at the endpoints because \(G(0)=G(1)=0\).}
The endpoint growth assumptions on \(G\), together with Hardy's inequality, establishes the existence of a solution to this problem in the form of an eigenpair $(\beta,V)$ where $\beta$ is a scalar eigenvalue and $V$ is the pool value function of interest \citet{zettl2005sturm}. The theorem asserts existence, not uniqueness, of a solution; some volatility profiles can have a family of admissible eigenpairs. Further regularity conditions on $G$ or truncated boundary conditions at the endpoints (making the problem non-singular) can guarantee uniqueness; we leave these characterizations outside the scope of this paper.

% The key object linking the stochastic process and the AMM is the second-order linear differential equation
% \begin{gather*}
%     \beta\cdot V(p)+G(p)^2\cdot V''(p) = 0, \\[0.25\baselineskip]
%     V(0) = V(1) = 0.
% \end{gather*}
% This boundary value problem (BVP) reveals that curvature of a uniform pool-value function scales inversely with the square of volatility. Wherever $G(p)$ is large, i.e.\ beliefs fluctuate rapidly, the function $V$ must bend more sharply to offset the higher quadratic variation of the price process. Conversely, where the process is relatively stable, less curvature is required. The BVP thus reveals that uniformity is a dynamic compensation mechanism between diffusion intensity and pricing function curvature.

Our assumptions on $G$ also produce two convenient properties of uniform pool-value functions $V$ for the given price dynamics: $V$ is concave and unimodal, both of which are natural for pool-value functions in the prediction market setting. % This result establishes existence of uniform pool-value functions without requiring a specific parametric form of $G$, highlighting that uniform AMMs are not tied to special processes such as Brownian bridges or Gaussian score dynamics, but rather exist for a broad class of win-martingales.
\Cref{cor:scale} further clarifies the structure of $V$ by showing that PM-AMMs are invariant to volatility scaling. When the volatility of the fiar price process is scaled by a multiple $c>0$, the same pool-value function V remains uniform; only the proportionality constant adjusts. Hence, uniform AMM shapes depends on the relative distribution of volatility across price states, not on its absolute magnitude: scaling overall uncertainty changes the rate of LVR but not where liquidity must concentrate to achieve uniform LVR.

\begin{corollary}[Scale invariance of uniform AMMs]\label{cor:scale}
Suppose that $\{P_t\}$ is a win-martingale controlled by $G$ satisfying Assumption \ref{ass:win-martingale}. Let $\tilde G(p)=c\cdot G(p)$ for some $c>0$. If $(\beta,V)$ solves the uniformity BVP for $G$, then $(c^2 \beta,V)$ solve the uniformity BVP for $\tilde G$.
\end{corollary}

% \begin{proposition}[Constructing a pricing function from a pool-value function]\label{prop:PVtoPF}\todo{this should go into the AMM section}
% Suppose that $V$ is continuous, twice-differentiable, and concave on $[0,1]$ with $V(0)=V(1)=0$. Then there exists a pricing function $F:\R_+^2\to\R_+$ with pool-value function $V$. \textcolor{red}{Check/relate to other papers.} \textcolor{blue}{Equivalent to those other papers (Frongillo et al., Angeris et al.)}
% \end{proposition}

% Proposition \ref{prop:PVtoPF} completes the constructive program by moving from the pool value representation back to an implementable AMM invariant. While the uniformity condition is naturally expressed in terms of the pool-value function $V(p)$, actual on-chain market makers are typically defined by a reserve-based pricing function $F(x,y)$. This result shows that any sufficiently regular, concave V satisfying the boundary conditions can be ``lifted'' to such an invariant.

Economically, this means that the pool-value function is the fundamental design object. Once one specifies how LP value should vary with probability, subject to concavity and vanishing at the boundaries, there always exists a reserve-based mechanism that implements it. The concavity of $V$ guarantees well-posed arbitrage problems and interior supporting reserves, ensuring that the resulting $F$ satisfies the monotonicity and curvature properties required for a valid AMM.

Thus a calibrated martingale model directly determines the AMM curvature required for uniform LVR. We now construct uniform AMMs for two example win-martingales.

\subsection{PM-AMMs for Volatility Profile Examples}

We now characterize and solve, to the extent possible, the uniformity BVPs for each of the six volatility profiles in Section \ref{sec:wm-example}. In each case, once a positive concave solution \(V\) is obtained (with unit liquidity), the corresponding uniform-LVR AMM is recovered from the demand functions, from which an invariant representation may be derived.

\vspace{1em}

\paraheader{Absorbed Brownian motion.}
The uniformity BVP is
\[
V''(p)+\beta\cdot V(p)=0,\qquad V(0)=V(1)=0, \qquad V(p)>0 \ \forall \ p\in(0,1).
\]
An eigenpair solution for $L=1$ is
\[
V(p)=\sin(\pi p),\qquad \beta=\pi^2.
\]
The demand functions
\[
X(p)=\sin(\pi p)+\pi(1-p)\cos(\pi p),\qquad
Y(p)=\sin(\pi p)-\pi p\cos(\pi p).
\]
After a series of substitutions, we can show that an invariant function for this AMM is
\[
F(x,y;L)
:=
x+y-2L \left(\sqrt{1-u^2}+u\arcsin(u)\right)=0,
\]
where $u=(L\pi)^{-1}(x-y)$.

\vspace{1em}

\paraheader{Wright-Fisher / Jacobi diffusion.}
The uniformity BVP is
\[
p(1-p)\cdot V''(p)+\beta \cdot V(p)=0,\qquad V(0)=V(1)=0\qquad V(p)>0 \ \forall \ p\in(0,1).
\]
An eigenpair solution for $L=1$ is
\[
V(p)=p(1-p),\qquad \beta=2.
\]
The demand functions induced by $V$ are
\begin{gather*}
X(p)=(1-p)^2, \qquad Y(p)=p^2.
\end{gather*}
implying that $\sqrt{X(p)}+\sqrt{Y(p)}=\sqrt{L }$, resulting in the invariant
\[
F(x,y)=\sqrt{x}+\sqrt{y}=\sqrt{L }.
\]
This invariant has been referred to in the literature as a \textit{constant elasticity market maker} with parameter $1/2$ \citep{frongillo2024axiomatic}.

\vspace{1em}

\paraheader{Logistic / log-odds diffusion.}
The uniformity BVP is
\[
p^2(1-p)^2 \cdot V''(p)+\beta \cdot V(p)=0,\qquad V(0)=V(1)=0\qquad V(p)>0 \ \forall \ p\in(0,1).
\]
An eigenpair solution for $L=1$ is
\[
V(p)=2 \sqrt{p(1-p)},\qquad \beta=\frac14.
\]
From Section \ref{sec:amm-example}, we know that this $V$ corresponds to the CPMM, so the CPMM is precisely the PM-AMM for a diffusion that is linear in log-odds space.

\vspace{1em}

% \paraheader{General Beta-type family.}
% The uniformity BVP is
% \[
% p^{2a}(1-p)^{2b}V''(p)+\beta\cdot V(p)=0,\qquad V(0)=V(1)=0.
% \]
% In general, this equation does not admit a closed-form solution. The special cases
% \[
% (a,b)=\Bigl(\frac12,\frac12\Bigr)
% \quad\text{and}\quad
% (a,b)=(1,1)
% \]
% recover the Wright-Fisher/CES and logistic/CPMM constructions above. More generally, the family \(p^a(1-p)^b\) yields a corresponding family of uniform AMMs that interpolate between these benchmarks and allow asymmetric boundary behavior when \(a\neq b\).

\paraheader{Gaussian score dynamics.}
The uniformity BVP is
\[
\phi\!\bigl(\Phi^{-1}(p)\bigr)^2 \cdot V''(p)+\beta \cdot V(p)=0,\qquad V(0)=V(1)=0\qquad V(p)>0 \ \forall \ p\in(0,1).
\]
An eigenpair solution for $L=1$ is
\[
V(p)=\phi(z),\qquad \beta=1.
\]
where \(z=\Phi^{-1}(p)\).
The demand functions induced by $V$ are
\[
X(p)= \phi(\Phi^{-1}(p))-(1-p)\cdot \Phi^{-1}(p),
\qquad
Y(p)= \phi(\Phi^{-1}(p))+p\cdot \Phi^{-1}(p).
\]
One way to express this as an invariant is as follows:
\begin{gather*}
    F(x,y) = (y-x)\cdot \Phi\left(\frac{y-x}{L }\right) + L \cdot\phi\left(\frac{y-x}{L }\right)-y=0.
\end{gather*}

\paraheader{The most exciting game.}
The uniformity BVP is
\[
    \frac{\sin^2(\pi p)}{\pi^2}\cdot V''(p)+\beta\cdot V(p)=0,
    \qquad
    V(0)=V(1)=0.
\]
An eigenpair with symmetric $V$ and $L=1$ is given by
\[
    V(p)
    =
    \sqrt{\sin(\pi p)}
    \left[
        \mathcal{P}_{-1/2}\!\left(\cos(\pi p)\right)
        +
        \mathcal{P}_{-1/2}\!\left(-\cos(\pi p)\right)
    \right],
    \qquad \beta = \frac{1}{4}.
\]
where \(\mathcal{P}_\nu\) denotes the Ferrers function of the first kind on
\((-1,1)\). The induced demand functions and invariant are omitted for brevity.

\subsection{From AMM Invariant to Win-Martingale}

For the converse direction of our correspondence, we begin with an existing pricing function and ask whether there is a compatible probability dynamics under which the AMM is uniform. \Cref{prop:PVtoWM} shows that, under appropriate curvature conditions on the induced pool-value function V, such a process always exists.

\begin{theorem}[Constructing a win-martingale that makes an AMM uniform]\label{prop:PVtoWM}
Suppose that $V:[0,1]\to\R_+$ is continuous, twice-differentiable on \((0,1)\), satisfies $V(0)=V(1)=0$, $V(p)>0$, and $V''(p)<0$ for $p\in(0,1)$. If $-V(p)/V''(p)$ is locally Lipschitz on \((0,1)\) and extends continuously to \([0,1]\) with endpoint values zero, then the process $\{P_t\}$ with dynamics
\begin{gather*}
    dP_t = \frac{G(P_t)}{h(t)}\,dW_t
\end{gather*}
where $\displaystyle G(P_t)=\sqrt{\beta\cdot\frac{V(P_t)}{-V''(P_t)}}$ for $\beta>0$ and $h(t)$ satisfies Assumption \ref{ass:win-martingale}(iv) is a win-martingale, and $(\beta,V)$ solves the uniformity BVP under the volatility profile given by $G$.
\end{theorem}

Once \(V\) is fixed, the compatible volatility profile \(G\) is pinned down by curvature: quadratic variation must offset \(V''\) so that instantaneous LVR is proportional to pool value.\footnote{Assumption~\ref{ass:win-martingale}(iii) is not required for the constructed process to be a win-martingale; the converse result only uses the clock condition in Assumption~\ref{ass:win-martingale}(iv), together with the stated regularity of the constructed volatility profile.} Together, \cref{prop:WMtoPV,prop:PVtoWM} give a bidirectional correspondence between AMM geometry and belief dynamics.

\subsection{Uniform Win-Martingales for AMM Examples}

\paraheader{Constant product market maker (CPMM).}
The CPMM is has pool-value function
\begin{gather*}
    V(p) = 2L \sqrt{p(1-p)}.
\end{gather*}
The ratio $-V(p)/V''(p)$ yields the dynamics for the process $\{P_t\}$ under which the CPMM has uniform LVR:
\begin{gather*}
    dP_t = \frac{P_t(1-P_t)}{h(t)}\,dW_t.
\end{gather*}

\paraheader{Logarithmic market scoring rule (LMSR).}
The LMSR is has pool-value function is
\begin{gather*}
    V(p) = L \cdot H(p).
\end{gather*}
where $H(p)=-(p\log p+(1-p)\log(1-p))$ is the entropy function. The ratio $-V(p)/V''(p)$ yields the dynamics for the process $\{P_t\}$ under which the LMSR has uniform LVR:
\begin{gather*}
    dP_t = \frac{\sqrt{P_t(1-P_t)\cdot H(P_t)}}{h(t)}\,dW_t.
\end{gather*}
Figure~\ref{fig:all} compares these canonical AMMs with the PM-AMMs constructed from the
Gaussian-score and most-exciting-game profiles.
%Additional depth plots appear in \cref{app:supp-figures}.

\begin{figure}[!htbp]
    \centering
    \includegraphics[width=0.95\linewidth]{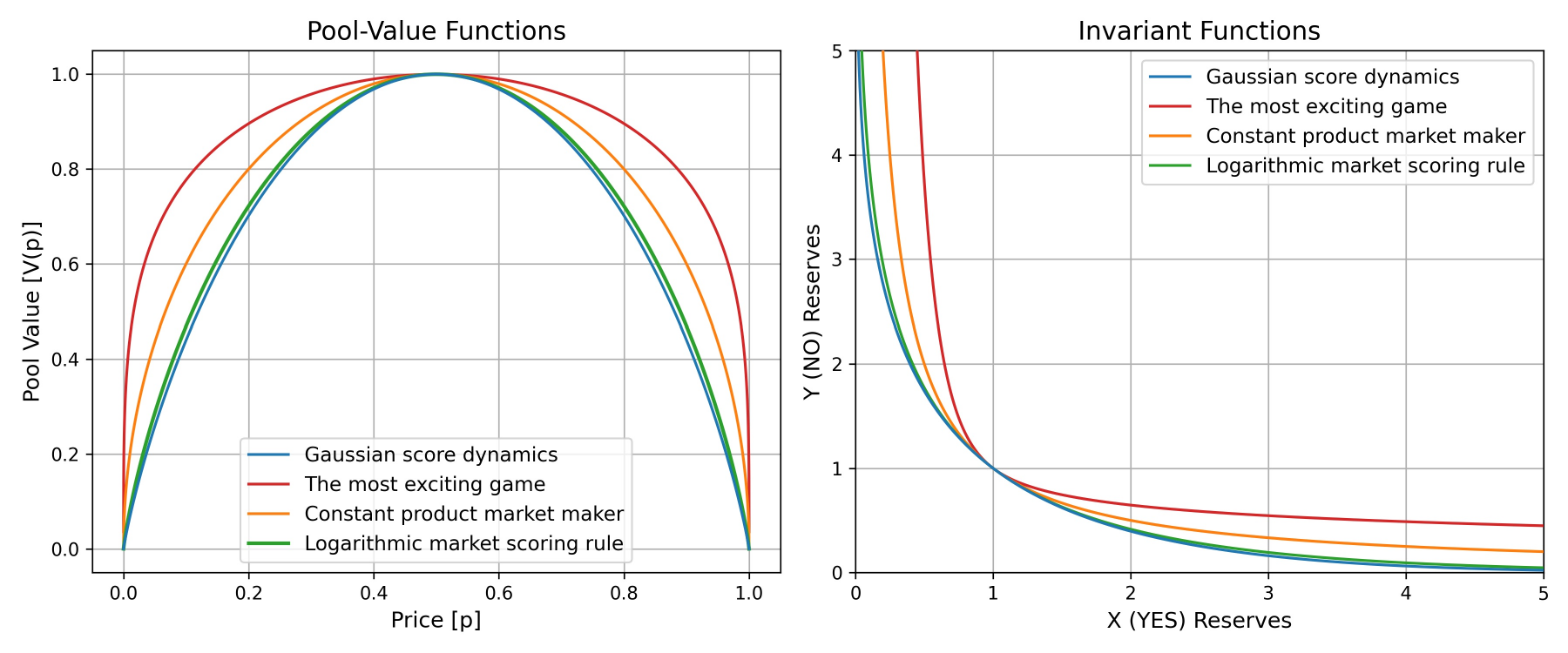}
    \includegraphics[width=0.95\linewidth]{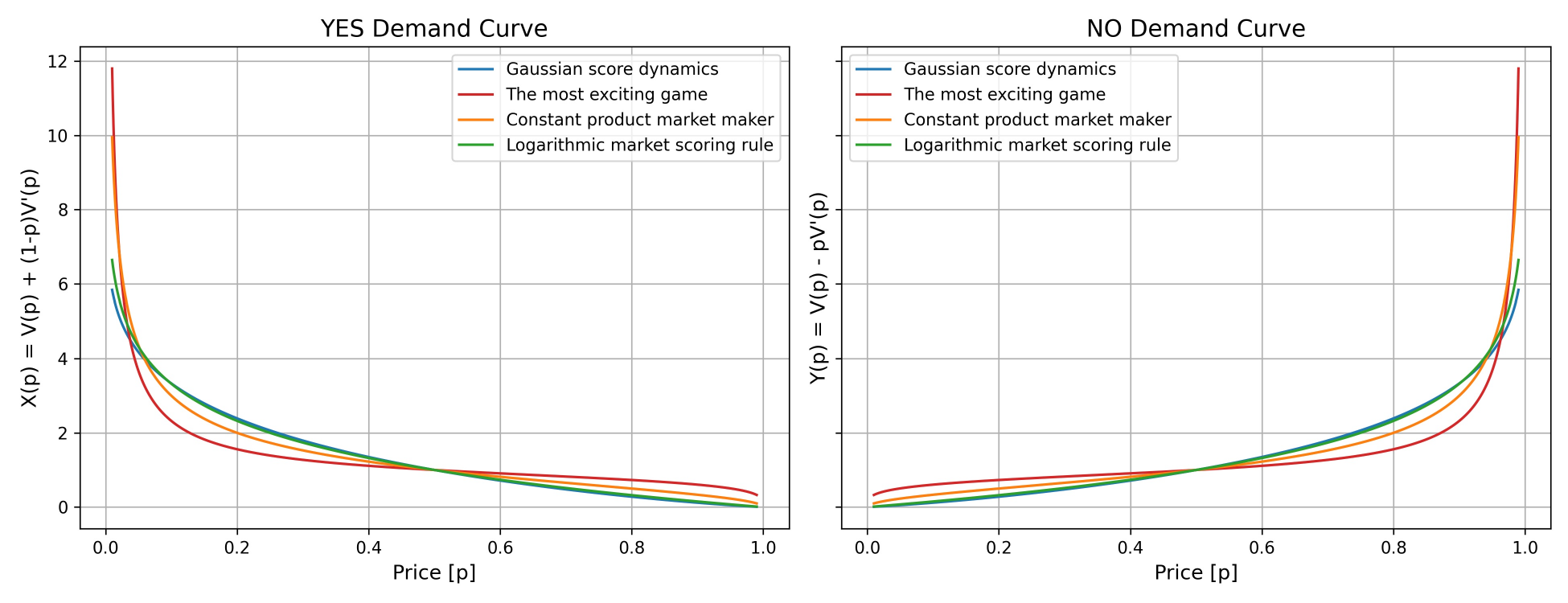}
    \caption{Comparison of the pool-value functions (top left), invariant functions (top right), YES-token demand function (bottom left), and NO-token demand function (bottom right) for the CPMM, LMSR, and the PM-AMMs for Gaussian score dynamics and ``the most exciting game.''}
    \label{fig:all}
\end{figure}

\begin{figure}[!htbp]
    \centering
    \includegraphics[width=0.475\linewidth]{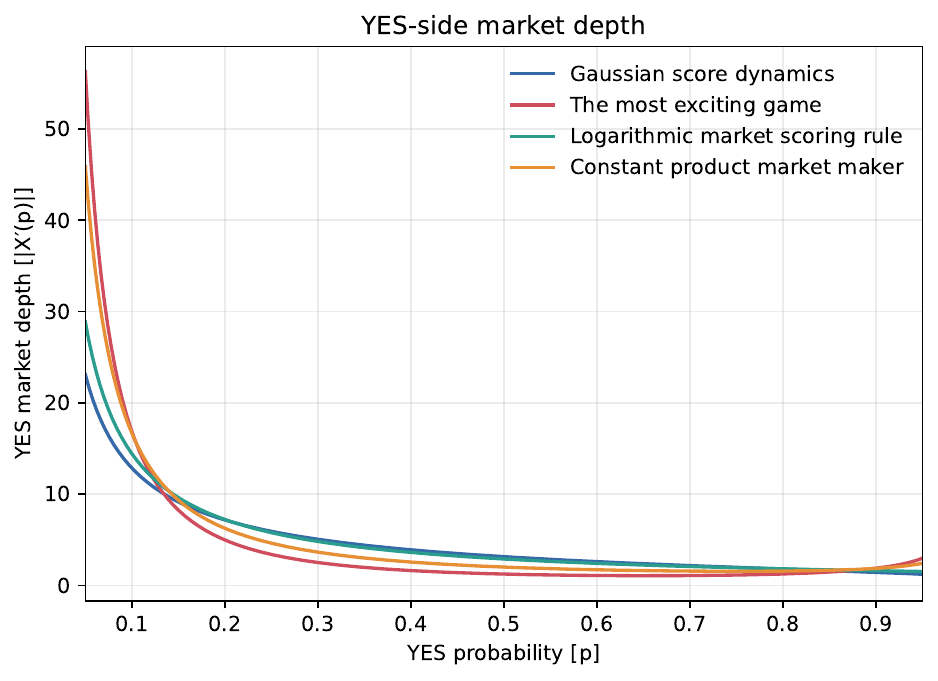}
    \includegraphics[width=0.475\linewidth]{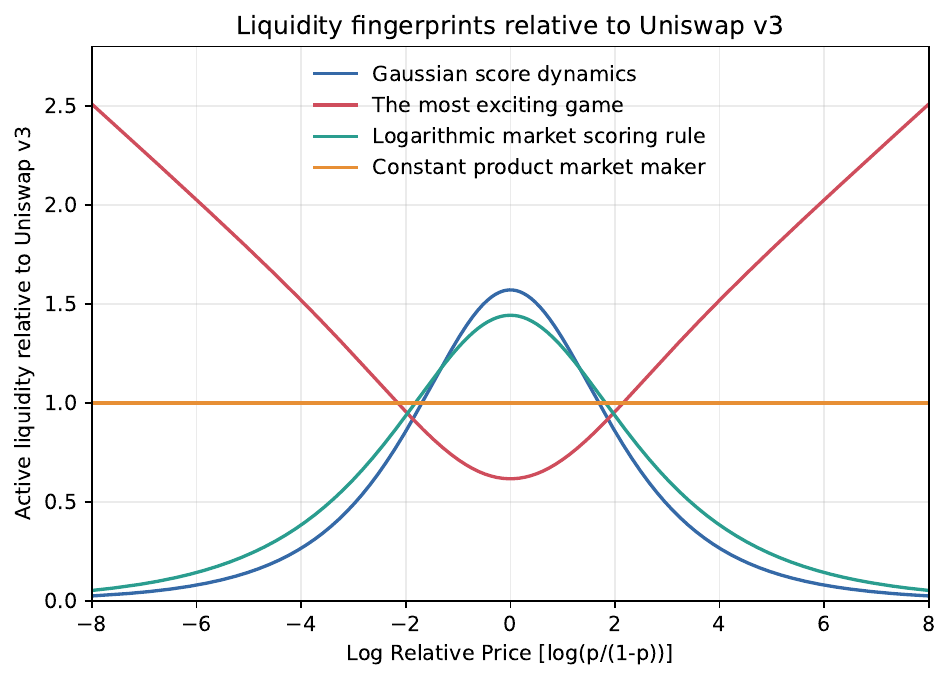}
    \caption{Left: absolute YES-side market depth, \(|X'(p)|\), measuring the local quantity of YES reserves available per unit change in the implied YES probability \(p\). Right: the same AMMs’ Uniswap-v3-equivalent liquidity fingerprints, expressed as active liquidity relative to the CPMM benchmark and plotted against log relative price, \(\log[p/(1-p)]\).}
    \label{fig:some}
\end{figure}

\subsection{Example: Violation of Uniformity}

The Gaussian score dynamics illustrate what goes wrong when the AMM curvature
is not matched to the price process. Suppose the fair price follows a volatility profile
\[
    G(p)=\phi(\Phi^{-1}(p)),
\]
but liquidity is provided through the CPMM pool-value function
\[
    V_{\mathrm{CPMM}}(p)=2L\sqrt{p(1-p)}.
\]
Then
\[
    V_{\mathrm{CPMM}}''(p)
    =
    -\frac{L}{2[p(1-p)]^{3/2}}.
\]
The instantaneous LVR rate as a fraction of pool value is therefore
\[
    \frac{\LVR_t}{V_{\mathrm{CPMM}}(P_t)}
    =
    -\frac{1}{2h(t)^2}
    \frac{V_{\mathrm{CPMM}}''(p)G(p)^2}{V_{\mathrm{CPMM}}(p)}
    =
    \frac{\phi(\Phi^{-1}(p))^2}
         {8h(t)^2p^2(1-p)^2}.
\]
This expression is not constant across price states. In fact, it explodes near
both endpoints. Let \(z=\Phi^{-1}(p)\). As \(p\downarrow0\), Mills' ratio gives
\[
    p=\Phi(z)\sim\frac{\phi(z)}{-z},
    \quad z\to-\infty \qquad \implies \qquad
    \frac{\phi(\Phi^{-1}(p))^2}{p^2(1-p)^2}
    \sim z^2
    \to\infty.
\]
Similarly, as \(p\uparrow1\), \(1-p\sim \phi(z)/z\) with \(z\to+\infty\), and
the same ratio diverges. Thus, under Gaussian score dynamics, the CPMM exposes
LPs to unbounded relative instantaneous LVR near prices close to zero or one,
whereas the Gaussian-score PM-AMM has relative instantaneous LVR equal to
\(1/(2h(t)^2)\) at every price state. The blow-up above is for LVR per unit of
pool value, which is the object controlled by uniformity.

\section{Dynamic PM-AMMs: Targeted Loss Over Time}\label{sec:dynamic}

The static theory chooses a PM-AMM invariant so that instantaneous LVR is uniform across price states. We now ask when those expected losses are incurred.

Once a uniform PM-AMM value function has been fixed, a deterministic liquidity schedule can implement any prescribed expected cumulative loss schedule. The same PM-AMM geometry can therefore be paired with front-loaded or back-loaded liquidity programs. Uniformity separates \emph{statewise} from \emph{timewise} risk: \(V\) fixes how losses scale across prices, while \(\{L_t\}\) determines how much exposure is deployed at each time. Without uniformity, changing liquidity would also distort the cross-sectional distribution of losses.

Given an AMM specified by its pool-value function $V$, let \(\{L_t\}_{0\le t<T}\) be a deterministic process denoting the LP's liquidity deployed into the AMM. Define the LP's marked-to-market pool value \(U_t\) and total wealth $W_t$ as follows:
\begin{gather*}
    U_t:=L_t V(P_t); \\[0.25\baselineskip]
    W_t:=c_0+U_t-\int_0^t \dot L_s\,V(P_s)\,ds,
\end{gather*}
where \(\dot L_s=dL_s/ds\), and the integral captures the cumulative cash transferred to or from the AMM as liquidity is adjusted over time. We interpret a target loss schedule as a function \(D\in C^1([0,T))\) such that $D(0)=0$, $D(t)\in[0,1]$ for all $t\in[0,T)$, and $D'(t)\geq0$ for all $t\in[0,T)$, where \(D(t)\) is the desired expected cumulative loss by time \(t\), expressed as a fraction of initial wealth. Since \(D\) is nondecreasing and bounded, the terminal limit \(D(T-):=\lim_{t\uparrow T}D(t)\) exists and satisfies \(D(T-)\le 1\).

\begin{theorem}[Implementing a target expected loss schedule]\label{thm:dynamic-target-loss}
Let \(w_0>0\) denote the LP's initial total wealth. Let \(D\in C^1([0,T))\) satisfy $D(0)=0$, $D(t)\in[0,1]$ for all $t\in[0,T)$, and $D'(t)\geq0$ for all $t\in[0,T)$. Under the liquidity schedule given by
\begin{equation*}
L_t
:=
\frac{2w_0}{\beta V(p_0)}\,h(t)^2
\exp\!\left(\int_0^t \frac{\beta}{2h(s)^2}\,ds\right)
D'(t),
\end{equation*}
with initial cash holdings $c_0:=w_0-L_0V(p_0)$, nonnegative initial cash feasibility requires
\[
    \frac{2h(0)^2D'(0)}{\beta}\le 1,
\]
and, whenever this condition holds,
\begin{enumerate}[(i)]
    \item $\E[W_t]=w_0(1-D(t))$ for all \(t<T\), i.e.\ the liquidity schedule $\{L_t\}$ implements the target expected cumulative loss schedule $D$;
    \item $\LVR_t = \frac{\beta}{2\cdot h(t)^2}\cdot L_t\cdot V(P_t)$, i.e.\ the instantaneous
      LVR for under the liquidity schedule is proportional to the LP position's value scaled by
      the information release rate.
\end{enumerate}
\end{theorem}

Thus, once the static uniform AMM has been fixed, the time profile of expected losses is controlled by the liquidity schedule while preserving statewise uniformity. At each time \(t\),
\[
    \LVR_t
    =
    \frac{\beta}{2h(t)^2}L_tV(P_t),
\]
so losses remain proportional to the current marked-to-market value of the LP's position.

\begin{corollary}[Withdrawal-only implementation]\label{cor:general-withdrawal-only}
If the map given by
\[
t\mapsto h(t)^2\exp\!\left(\int_0^t \frac{\beta}{2h(s)^2}\,ds\right)D'(t)
\]
is nonincreasing in $t$ for $t\in[0,T)$ and
\[
c_0=w_0-L_0 V(p_0)\ge 0.
\]
Then the liquidity schedule $\{L_t\}$ is nonincreasing, so the target expected-loss schedule is implementable by deterministic liquidity withdrawals only without external recapitalization.
\end{corollary}

These conditions identify when a target loss schedule is compatible with one-way de-risking rather than later recapitalization.

\paraheader{Square-root clock.}
For the clock \(h(t)=\sqrt{T-t}\), the liquidity schedule simplifies to
\[
    L_t
    =
    \frac{2w_0T^{\beta/2}}{\beta V(p_0)}
    (T-t)^{1-\beta/2}D'(t),
\]
so a linear expected-loss schedule has \(L_t\propto(T-t)^{1-\beta/2}\). Hence \(\beta<2\) corresponds to gradual withdrawal, \(\beta=2\) to constant active liquidity, and \(\beta>2\) to liquidity that grows as \(t\uparrow T\). For any volatility profile, replacing \(G\) by \(cG\) rescales \(\beta\) by \(c^2\). 

\paraheader{Fully-worked Gaussian score example.} Consider the Gaussian score dynamics ,
\[
    dP_t = \frac{\phi(\Phi^{-1}(p))}{\sqrt{T-t}}
\]
The corresponding uniform PM-AMM has
\[
    V(p)=\phi(\Phi^{-1}(p)),
    \qquad
    \beta=1.
\]
Substituting \(\beta=1\) into the square-root clock liquidity rule gives
\[
    L_t
    =
    \frac{2w_0\sqrt{T}}{V(p_0)}
    \sqrt{T-t}\,D'(t),
    \qquad 0\le t<T.
\]
Thus the liquidity schedule is proportional to \(D'(t)\sqrt{T-t}\): the square-root clock mechanically pulls liquidity down as resolution approaches, while the target loss schedule determines how much loss is intentionally allocated at each time.

\paraheader{Constant expected-loss rate.}
First suppose that expected losses accumulate at a constant rate from \(D(0)=0\) to terminal loss \(D(T-)=1/4\):
\[
    D(t)=\frac{t}{4T},
    \qquad
    D'(t)=\frac{1}{4T}.
\]
The liquidity rule becomes
\[
    L_t
    =
    \frac{w_0}{2V(p_0)}
    \sqrt{\frac{T-t}{T}}.
\]
This schedule withdraws liquidity smoothly over time. The expected loss rate is constant by construction, but the active liquidity as information is released faster closer to resolution.

\paraheader{Front-loaded expected-loss schedule.}
Now consider
\[
    D(t)=\frac{t}{2T+2t}, \qquad D'(t)=\frac{2T}{(2T+2t)^2}, 
    \qquad 0<t<T,
\]
with the continuous extension \(D(0)=0\). This schedule also satisfies \(D(T-)=1/4\), but the loses are front-loaded, resulting  in the liquidity schedule
\[
    L_t
    =
    \frac{w_0T^{3/2}}{V(p_0)}
    \frac{\sqrt{T-t}}{(T+t)^2}.
\]

% A fully worked Gaussian-score example appears in \cref{app:supp-dynamic}.

%%% Local Variables:
%%% mode: LaTeX
%%% TeX-master: "main"
%%% End:

\section{Conclusion}

This paper develops a framework for designing AMMs for binary prediction markets by matching AMM geometry to the stochastic dynamics of the fair price. We show that uniform LVR can be achieved by a boundary value problem linking the pool-value function $V$ to the win-martingale volatility profile $G$, yielding a bidirectional correspondence between belief processes and AMM invariants. We then extend the framework dynamically, showing that once the uniform-LVR invariant is fixed, liquidity schedules can implement prescribed expected loss profiles over time. These insights can allow LPs and AMM designers to allocate liquidity strategically based on the underlying event to either smooth expected losses or subsidize price discovery. Future directions include empirically calibrating \(G\) and \(h\) from prediction market data, fitting observed price paths to canonical win-martingales, or extending the theory to multi-outcome and combinatorial markets.

\clearpage

\bibliography{uniform-lvr}

\newpage

\appendix

% \section{Supplementary Material}\label{app:supp}

\section{Supplementary Derivations for Price Processes}\label{app:supp-price-processes}

\paraheader{Wright-Fisher / Jacobi diffusion.}
The standard population-genetic derivation starts from a finite diploid population with two alleles at a fixed locus. If the current allele-$A$ frequency is \(p\), the next generation is obtained by drawing \(2N\) genes from the parental gene pool, so the one-step frequency has binomial sampling variance \(p(1-p)/(2N)\). Accelerating generation time by \(2N\) and letting \(N\to\infty\) turns random genetic drift into a diffusion with generator \(\frac12 p(1-p)f''(p)\) in the neutral case, meaning absent selection and mutation. Equivalently, the limiting process satisfies \(dP_t=\sqrt{P_t(1-P_t)}\,dW_t\). This genetics story is distinct from the log-likelihood signal limit below: here the square-root volatility comes from sampling noise in population frequencies rather than additive evidence in log-odds.

\paraheader{Logistic / log-odds diffusion.}
For the logistic profile \(G(p)=p(1-p)\), define the log-odds process
\[
    \mathcal{O}_t := \log\frac{P_t}{1-P_t}.
\]
Under a separable win-martingale with this volatility profile and information clock \(h\), It\^o's formula yields
\[
    d\mathcal{O}_t
    =
    \frac{1}{h(t)}\,dW_t
    +
    \frac{P_t-\frac12}{h(t)^2}\,dt.
\]
Thus the log-odds process has constant instantaneous volatility, up to the information clock \(h(t)\). Equivalently, under informational time \(d\tau = h(t)^{-2}\,dt\),
\[
    d\mathcal{O}_\tau
    =
    dB_\tau
    +
    \left(\frac{e^{\mathcal{O}_\tau}}{1+e^{\mathcal{O}_\tau}}-\frac12\right)d\tau.
\]
The drift term above results from the fact that \(P_t\), rather than \(\mathcal{O}_t\), is a martingale.

This specification is natural in settings where information arrives through approximately additive increments to log-likelihood ratios. Suppose the terminal outcome is a latent binary state \(\theta\in\{0,1\}\) and the market observes conditionally independent binary signals. Conditional on the state, let
\[
    \mathbb{P}(S_n=1\mid \theta=0)=\frac12-\varepsilon,
    \qquad
    \mathbb{P}(S_n=1\mid \theta=1)=\frac12+\varepsilon .
\]
Each signal contributes a log-likelihood-ratio increment of order \(\varepsilon\); when signals arrive at rate \(\varepsilon^{-2}\), the log-odds process has a nondegenerate diffusive limit. Mapping back from log-odds to probabilities through the logistic function produces a probability diffusion with volatility proportional to \(p(1-p)\), up to an overall time scale. More generally, Bayes' rule updates posterior odds multiplicatively:
\[
    \frac{\mathbb{P}(\theta=1\mid \mathcal{F}_{n+1})}
         {\mathbb{P}(\theta=0\mid \mathcal{F}_{n+1})}
    =
    \frac{\mathbb{P}(\theta=1\mid \mathcal{F}_{n})}
         {\mathbb{P}(\theta=0\mid \mathcal{F}_{n})}
    \cdot
    \frac{\mathbb{P}(S_{n+1}\mid \theta=1)}
         {\mathbb{P}(S_{n+1}\mid \theta=0)}.
\]
Taking logarithms shows that each signal contributes additively to the log-odds:
\[
    \mathcal{O}_{n+1}
    =
    \mathcal{O}_n
    +
    \log
    \frac{\mathbb{P}(S_{n+1}\mid \theta=1)}
         {\mathbb{P}(S_{n+1}\mid \theta=0)}.
\]
The logistic diffusion is the continuous-time analogue of this mechanism, where evidence is approximately additive in log-likelihood-ratio space, while prices are obtained by mapping log-odds back to probabilities through the logistic function.

\paraheader{Gaussian score dynamics.}
Let a latent fundamental \(Z_t\) represent the score differential between two teams, candidates, or states, with positive values favoring the YES token. Suppose
\[
    dZ_t = \sigma\,dW_t,
\]
and the YES token pays off if the terminal score differential is nonnegative. Then
\[
    P_t
    =
    \E[\1\{Z_T\geq0\}\mid \mathcal{F}_t]
    =
    \mathbb{P}(Z_T\geq 0\mid \mathcal{F}_t)
    =
    \Phi\!\left(\frac{Z_t}{\sigma\sqrt{T-t}}\right),
\]
where \(\Phi\) denotes the standard normal distribution function. By It\^o's formula,
\[
    dP_t
    =
    \frac{\phi(\Phi^{-1}(P_t))}{\sqrt{T-t}}\,dW_t,
\]
where \(\phi\) is the standard normal density. Thus
\[
    G(p)=\phi(\Phi^{-1}(p)),
    \qquad
    h(t)=\sqrt{T-t}.
\]
The scale parameter \(\sigma\) drops out: increasing the volatility of the latent score also increases the posterior uncertainty about \(Z_T\) by the same factor, leaving the probability dynamics expressed in terms of \(P_t\) unchanged.

\paraheader{Brody-Hughston-Macrina clock.}
In the information-based asset-pricing framework of \citet{brody2008information}, the market observes an information process of the form
\[
    \xi_t = \sigma t X_T + B_t^{(T)},
    \qquad 0\leq t\leq T,
\]
where \(X_T\) is the terminal payoff and \(\{B_t^{(T)}\}_{0\leq t\leq T}\) is a Brownian bridge satisfying \(B_0^{(T)}=B_T^{(T)}=0\). The first term represents accumulating signals about the terminal payoff, while the Brownian bridge represents noise that is pinned down to zero at resolution. For a binary payoff \(X_T\in\{0,1\}\), the fair price is the conditional expectation
\[
    P_t = \mathbb{E}[X_T\mid \xi_t].
\]
Bayes' rule implies that the posterior log-odds are affine in the observed signal:
\[
    \mathcal{O}_t
    =
    \log\frac{P_t}{1-P_t}
    =
    \log\frac{P_0}{1-P_0}
    +
    \frac{T}{T-t}\left(\sigma \xi_t-\frac12\sigma^2 t\right).
\]
Thus the BHM information process generates belief dynamics that are naturally additive in log-likelihood-ratio space. In the binary case, the induced posterior process satisfies a separable diffusion of the logistic form
\[
    dP_t
    =
    \frac{\sigma T}{T-t} \cdot P_t(1-P_t)\,dW_t.
\]
Equivalently, in our notation,
\[
    G(p)=p(1-p),
    \qquad
    h(t)=\frac{T-t}{\sigma T}.
\]

\section{Proofs of Results}

\subsection{Proof of \cref{prop:separable-win-martingale}}

\begin{proof}
For (i), let
\[
        \tau(t):=\int_0^t h(s)^{-2}\,ds,\qquad 0\le t<T .
\]
By Assumption~3(iv), $\tau(t)<\infty$ for every $t<T$ and
$\tau(t)\nearrow\infty$ as $t\nearrow T$. Let $\{Q_t\}$ be the solution on $[0,\infty)$ of the absorbed diffusion
\[
        dQ_s=G(Q_s)\,dB_s,\qquad Q_0=p_0\in(0,1),
\]
where $\{B_t\}$ is a standard Brownian motion (in informational time) with natural filtration $\{\F_t^B\}$. The process is stopped and kept fixed after the first hitting time
\[
        \zeta:=\inf\{s\ge0:Q_s\in\{0,1\}\}.
\]
Since $G$ is locally Lipschitz on $(0,1)$, strictly positive on $(0,1)$, and satisfies $G(0)=G(1)=0$, we have strong existence and pathwise uniqueness up to $\zeta$, and the absorbed extension
is well-defined on all of $\tau\in[0,\infty)$.

The process $\{Q_t\}$ is a continuous local martingale taking values in $[0,1]$. Since it is bounded, it is a true martingale. Hence $Q_s$ converges almost surely and in $L^1$ to a limit $Q_\infty:=\lim_{s\to\infty} Q_s$. We claim that $Q_\infty\in\{0,1\}$ almost surely. Suppose, to the contrary, that $A:=\{Q_\infty\in(0,1)\}$ has positive probability. Then on $A$, there exist random constants $\varepsilon>0$ and $S<\infty$ such that
\[
        Q_s\in[\varepsilon,1-\varepsilon]\qquad\text{for all }s\ge S .
\]
Since $G$ is continuous and strictly positive on $[\varepsilon,1-\varepsilon]$, there is a constant $c_\varepsilon>0$ such that $G(x)^2\ge c_\varepsilon$ on this interval. Then on $A$, we have
\[
        \langle Q\rangle_\infty
        =
        \int_0^\infty G(Q_s)^2\,ds
        \ge
        \int_S^\infty c_\varepsilon\,ds
        =
        \infty .
\]
By the Dambis-Dubins-Schwarz theorem, a continuous local martingale with
quadratic variation diverging to infinity can be written as a Brownian
motion run for infinite time. Such a process cannot converge to a finite
limit, contradicting the assumption that \(Q_s\to Q_\infty\in(0,1)\) on the
event \(A\). Thus, $Q_\infty\in\{0,1\}$ almost surely.

Define $P_t:=Q_{\tau(t)}$ for $t\in[0,T]$ and $P_T\coloneqq Q_\infty$. By a deterministic time change,
\[
        dP_t
        =
        \frac{G(P_t)}{h(t)}\,dW_t
\]
for a Brownian motion $W$ with respect to the time-changed filtration. Since $\{Q_t\}$ is a bounded martingale, $\{P_t\}$ is a bounded martingale on every compact subinterval of $[0,T)$, and by the identity $Q_s=\mathbb E[Q_\infty\mid\mathcal F_s^B]$, we have
\[
        P_t
        =
        Q_{\tau(t)}
        =
        \mathbb E[Q_\infty\mid\mathcal F_{\tau(t)}^B]
        =
        \mathbb E[P_T\mid\mathcal F_t],
\]
for $t\in[0,T)$. Thus $\{P_t\}$ is a martingale on $[0,T]$ taking values in $[0,1]$, and satisfies $P_T\in\{0,1\}$ almost surely. Thus, $\{P_t\}$ is a win-martingale.

For (ii), since $P_T\in\{0,1\}$ almost surely, $P_T^2=P_T$, and the martingale property gives $\mathbb{E}[P_T\mid\mathcal{F}_t]=P_t$. Hence
\[
    \Var[P_T\mid\mathcal{F}_t]
    =
    \mathbb{E}[P_T^2\mid\mathcal{F}_t] - P_t^2
    =
    P_t - P_t^2
    =
    P_t(1-P_t).
\]
By It\^o isometry applied to the stochastic integral $P_T - P_t = \int_t^T (G(P_s)/h(s))\,dW_s$,
\[
    \mathbb{E}\!\left[(P_T-P_t)^2\,\middle|\,\mathcal{F}_t\right]
    =
    \mathbb{E}\!\left[\int_t^T \frac{G(P_s)^2}{h(s)^2}\,ds\,\middle|\,\mathcal{F}_t\right].
\]
Combining the two identities yields the claim.
\end{proof}

\subsection{Proof of \cref{prop:pool-value-well-posed-main}}

Fix $L>0$ and additionally assume that the level set
\[
        \Gamma_L:=\{(x,y)\in\mathbb R_{++}^2:F(x,y;L)=0\}
\]
is nonempty. Write the graph of $\Gamma_L$ as
\[
        \Gamma_L=\{(x,f_L(x)):x\in(0,\infty)\},
\]
i.e.\ $y=f(x)$. By the Implicit Function Theorem,
\[
        f_L'(x)=-\frac{F_x(x,f_L(x);L)}{F_y(x,f_L(x);L)} .
\]
By the Inada-like conditions on $F$, we also have $\lim_{x\searrow0}f'_L(x)=-\infty$ and $\lim_{x\nearrow}f'_L(x)=0$. Differentiating $f'_L$ again with respect to \(x\) yields
\[
        f_L''(x)
        =
        -\frac{F_{xx}F_y^2-2F_{xy}F_xF_y+F_{yy}F_x^2}{F_y^3}
        \bigg|_{(x,f_L(x);L)} .
\]
Since $F_y>0$ and $F_y^2F_{xx}-2F_xF_yF_{xy}+F_x^2F_{yy}<0$ by assumption, it follows that $f''_L(x)>0$, i.e.\ $f_L$ is strictly convex. For $p\in(0,1)$, write the objective as
\[
        \phi_p(x):=px+(1-p)\cdot f_L(x),\qquad x>0 .
\]
Then $\phi_p''(x)=(1-p)\cdot f_L''(x)>0$, so $\phi_p$ is strictly convex. Moreover, $\phi_p'(x)=p+(1-p)f_L'(x)$, so $\lim_{x\searrow}\phi_p'(x)=-\infty$ and $\lim_{x\nearrow\infty}\phi_p'(x)=p>0$. Thus, there exists a unique $x(p)\in(0,\infty)$ satisfying $phi_p'(x(p))=0$. Strict convexity implies that this point is the unique global minimizer of
$\phi_p$. Setting
\[
        X(p):=x(p),\qquad Y(p):=f_L(x(p)),
\]
we obtain the unique solution of the pool-value problem
\[
        V_L(p)
        =
        \inf_{(x,y)\in\Gamma_L}\{px+(1-p)y\}
        =
        p \cdot X(p)+(1-p) \cdot Y(p).
\]
The first-order condition is
\[
        p+(1-p)f_L'(X(p))=0 \iff \frac{F_x(X(p),Y(p);L)}
             {F_y(X(p),Y(p);L)}
        =
        \frac{p}{1-p},
\]
Since $f_L''>0$, the Implicit Function Theorem applied to
$p+(1-p)\cdot f_L'(x)=0$ implies that $X$ is $C^1$ on $(0,1)$, and hence so are
$Y$ and $V_L$.

For concavity, note that $V_L$ is the pointwise infimum over $\Gamma_L$ of
the affine functions
\[
        p\mapsto px+(1-p)y .
\]
Therefore $V_L$ is concave on $(0,1)$.

Finally, the Inada-like conditions imply that $\lim_{p\downarrow0} V_L(p)=\inf_{(x,y)\in\Gamma_L} y=0$ and $\qquad\lim_{p\uparrow1} V_L(p) = \inf_{(x,y)\in\Gamma_L} x=0$, so $V_L$ extends continuously to $[0,1]$ with $V_L(0)=V_L(1)=0$

\subsection{Proof of \cref{prop:demand-functions}}

Assume the hypotheses of Proposition~6, and let
\[
        V(p)=p\cdot X(p)+(1-p) \cdot Y(p)
\]
be the pool-value function at the unique optimizer $(X(p),Y(p))$. By the envelope theorem, we have $V'(p)=X(p)-Y(p)$, yielding a linear system in $X(p)$ and $Y(p)$ with solution
\[
        X(p)=V(p)+(1-p)\cdot V'(p),
        \qquad
        Y(p)=V(p)-p \cdot V'(p).
\]
This proves the demand-function identities.

Conversely, let $V:[0,1]\to\mathbb R_+$ be concave, continuous, differentiable on \((0,1)\), positive
on $(0,1)$, and satisfy $V(0)=V(1)=0$. For $L>0$, define
\[
        F(x,y;L)
        :=
        \inf_{q\in(0,1)}
        \frac{qx+(1-q)y}{V(q)}
        -
        L .
\]
We claim that the level set $F(x,y;L)=0$ has pool-value function $LV$.

Fix $p\in(0,1)$. By concavity of $V$, for every $q\in(0,1)$, we have
\[
        V(q)\le V(p)+V'(p)(q-p).
\]
Using the demand functions
\[
        X(p)=L(V(p)+(1-p)\cdot V'(p)),
        \qquad
        Y(p)=L(V(p)-p\cdot V'(p)),
\]
we have
\[
        q\cdot X(p)+(1-q) \cdot Y(p)
        =
        L(V(p)+(q-p)\cdot V'(p))
        \ge
        LV(q) \implies \frac{qX(p)+(1-q)Y(p)}{V(q)}\ge L.
\]
for all $q\in(0,1)$ with equality at $q=p$, so $F(X(p),Y(p);L)=0$. Moreover, $p\cdot X(p)+(1-p)\cdot Y(p)=L\cdot V(p)$, so $V_L(p)\le LV(p)$.

On the other hand, $F(x,y;L)=0$, so $qx+(1-q)y\ge L\cdot V(q)$ for all $q\in(0,1)$. Taking $q=p$ yields $px+(1-p)y\ge LV(p)$, so $V_L(p)\ge LV(p)$. Thus $V_L(p) = L \cdot V(p)$for all $p\in(0,1)$. 
\[
        V_L(p)=LV(p),\qquad p\in(0,1).
\]
Thus the reconstructed invariant is equivalent to the pool-value function $LV$. This converse reconstruction proves the pool-value identity only. Without stronger regularity and strict-concavity assumptions on \(V\), the infimum representation need not satisfy the \(C^2\), sign, curvature, or boundary conditions in Assumption~\ref{ass:pricing-function}.

\subsection{Proof of \cref{prop:lvr-decomposition}}

\begin{proof}
For (i), apply It\^o's formula to $V(P_t)$ using $dP_t = G(P_t)/h(t)\,dW_t$ and $d\langle P\rangle_t = G(P_t)^2/h(t)^2\,dt$:
\[
    dV(P_t)
    =
    V'(P_t)\,\frac{G(P_t)}{h(t)}\,dW_t
    +
    \frac{1}{2}V''(P_t)\,\frac{G(P_t)^2}{h(t)^2}\,dt
    =
    V'(P_t)\,\frac{G(P_t)}{h(t)}\,dW_t
    - \LVR_t\,dt.
\]
Nonnegativity of $\LVR_t$ follows from concavity of $V$.

For (ii), fix $t\in[0,T]$ and integrate (i) over $[0,t]$:
\[
    V(P_t) - V(p_0) + \int_0^t \LVR_s\,ds
    =
    \int_0^t V'(P_s)\,\frac{G(P_s)}{h(s)}\,dW_s.
\]
The right-hand side is a continuous local martingale. Let $\{\tau_n\}$ be a localizing sequence with $\tau_n\nearrow\infty$, so the stopped process is a true martingale with zero expectation at $t\wedge\tau_n$:
\[
    \E[V(P_{t\wedge\tau_n})] - V(p_0) + \E\!\left[\int_0^{t\wedge\tau_n}\LVR_s\,ds\right] = 0.
\]
As $n\to\infty$, $V(P_{t\wedge\tau_n})\to V(P_t)$ a.s.\ by path continuity, and the convergence holds in expectation by bounded convergence since $V$ is continuous on $[0,1]$ and therefore bounded; $\int_0^{t\wedge\tau_n}\LVR_s\,ds\nearrow\int_0^t\LVR_s\,ds$ by monotone convergence. This yields $V(p_0) - \E[V(P_t)] = \E\!\left[\int_0^t\LVR_s\,ds\right]$.
\end{proof}

\subsection{Proof of \cref{prop:WMtoPV}}

\begin{lemma}
\label{lem:hardy-critical-ground-state}
Let $w:(0,1)\to(0,\infty)$ be locally bounded. Suppose there exist constants
$C_0,C_1>0$ and $\delta\in(0,1/2)$ such that
\[
    w(p)\le \frac{C_0}{p^2},
    \qquad p\in(0,\delta),
\]
and
\[
    w(p)\le \frac{C_1}{(1-p)^2},
    \qquad p\in(1-\delta,1).
\]
Define the quadratic forms
\[
    Q_\lambda[\phi]
    :=
    \int_0^1 |\phi'(p)|^2\,dp
    -
    \lambda\int_0^1 w(p)\phi(p)^2\,dp,
    \qquad \phi\in C_c^\infty(0,1),
\]
and set
\[
    \beta_\ast
    :=
    \sup\left\{
        \lambda\ge 0:
        Q_\lambda[\phi]\ge 0
        \text{ for all }\phi\in C_c^\infty(0,1)
    \right\}.
\]
Then $0<\beta_\ast<\infty$, and there exists a function
\[
    V\in C([0,1])\cap C^2((0,1))
\]
such that
\[
    V(0)=V(1)=0,
    \qquad
    V(p)>0\quad\text{for }p\in(0,1),
\]
and
\[
    -V''(p)=\beta_\ast w(p)V(p),
    \qquad p\in(0,1).
\]
\end{lemma}

\begin{proof}
The endpoint assumptions and Hardy's inequality imply that there exists
$C>0$ such that
\[
    \int_0^1 w(p)\phi(p)^2\,dp
    \le
    C\int_0^1 |\phi'(p)|^2\,dp,
    \qquad \phi\in C_c^\infty(0,1).
\]
Indeed, near $0$ we use
\[
    \int_0^\delta \frac{\phi(p)^2}{p^2}\,dp
    \le
    4\int_0^\delta |\phi'(p)|^2\,dp,
\]
and near $1$ we use
\[
    \int_{1-\delta}^1 \frac{\phi(p)^2}{(1-p)^2}\,dp
    \le
    4\int_{1-\delta}^1 |\phi'(p)|^2\,dp.
\]
On $[\delta,1-\delta]$, local boundedness of $w$ and the usual
one-dimensional Poincare inequality give the remaining bound. Therefore
\[
    \beta_\ast\ge C^{-1}>0.
\]

Also, $\beta_\ast<\infty$. To see this, choose any nonzero
$\eta\in C_c^\infty(0,1)$. Since $w>0$ on $(0,1)$,
\[
    \int_0^1 w(p)\eta(p)^2\,dp>0.
\]
If $Q_\lambda[\eta]\ge0$, then
\[
    \lambda
    \le
    \frac{\displaystyle\int_0^1 |\eta'(p)|^2\,dp}
         {\displaystyle\int_0^1 w(p)\eta(p)^2\,dp}.
\]
Hence $\beta_\ast<\infty$.

The Hardy bound above shows that the potential term is form-bounded with
respect to the Dirichlet energy. Hence each \(Q_\lambda\), and in particular
the critical form \(Q_{\beta_\ast}\), is closable on \(L^2(0,1)\). By the
ground-state alternative for second-order elliptic Schrödinger forms
\citep[see, e.g.,][]{pinchover1988criticality,devyver2014hardy}, applied to
the critical form \(Q_{\beta_\ast}\), the operator
\[
    -\frac{d^2}{dp^2}-\beta_\ast w(p)
\]
admits a positive ground state
\[
    V\in C^2((0,1)),
    \qquad
    V>0\quad\text{on }(0,1),
\]
satisfying
\[
    -V''(p)=\beta_\ast w(p)V(p),
    \qquad p\in(0,1).
\]
This is the standard criticality theorem:
at the critical value, the form is critical, admits a null sequence, and
the null sequence converges locally to the unique positive ground state.

It remains to record the endpoint behavior. The endpoint Hardy bounds put the
associated one-dimensional Sturm-Liouville expression in the lower
semibounded Friedrichs regime at both endpoints. In this regime the Friedrichs
ground state has principal-solution behavior at each endpoint. To see the
local behavior in the model case, near \(p=0\) the equation is comparable to
\[
    -V''(p)=\frac{c}{p^2}V(p),
    \qquad c>0,
\]
whose power-law solutions are \(p^{(1\pm\sqrt{1+4c})/2}\). The principal
solution is the one with exponent \((1+\sqrt{1+4c})/2\), and it vanishes at
the endpoint, while the nonprincipal solution has the singular endpoint
behavior excluded by the Friedrichs extension. The same argument applies at
\(p=1\), with \(p\) replaced by \(1-p\). Principal solutions at the left and
right endpoints are precisely the Friedrichs Dirichlet boundary
representatives for this singular problem. Consequently,
\[
    \lim_{p\searrow0}V(p)=0,
    \qquad
    \lim_{p\nearrow1}V(p)=0.
\]
Defining
\[
    V(0):=0,\qquad V(1):=0
\]
gives
\[
    V\in C([0,1])\cap C^2((0,1)).
\]
This proves the lemma.
\end{proof}

% \begin{theorem}[Constructing a uniform AMM for a win-martingale]
% Let $G$ satisfy Assumption~\ref{ass:win-martingale}. Then there exists a pair $(\beta,V)$ such that $\beta>0$ and $V\in C([0,1])\cap C^2((0,1))$ such that
% \[
%     \beta\cdot V(p)+G(p)^2 \cdot V''(p)=0,
%     \qquad p\in(0,1),
% \]
% and
% \[
%     V(0)=V(1)=0.
% \]
% Moreover, $V>0$ on $(0,1)$, and $V$ is concave and unimodal on $[0,1]$.
% Consequently, for every $L>0$, the pair $(\beta,LV)$ solves the same
% uniformity BVP.
% \end{theorem}

\begin{proof}
Set
\[
    w(p):=\frac{1}{G(p)^2},
    \qquad p\in(0,1).
\]
By Assumption~\ref{ass:win-martingale}(i), $G(p)>0$ for $p\in(0,1)$.
Hence $w$ is positive on $(0,1)$. Since $G$ is locally Lipschitz on
$(0,1)$, it is locally bounded away from zero on compact subintervals of
$(0,1)$. Therefore $w$ is locally bounded on $(0,1)$.

By Assumption~\ref{ass:win-martingale}(iii), there exist constants
$a_0,a_1>0$ and $\delta\in(0,1/2)$ such that
\[
    G(p)\ge a_0p,
    \qquad p\in(0,\delta),
\]
and
\[
    G(p)\ge a_1(1-p),
    \qquad p\in(1-\delta,1).
\]
Therefore
\[
    w(p)=\frac1{G(p)^2}
    \le
    \frac{1}{a_0^2p^2},
    \qquad p\in(0,\delta),
\]
and
\[
    w(p)=\frac1{G(p)^2}
    \le
    \frac{1}{a_1^2(1-p)^2},
    \qquad p\in(1-\delta,1).
\]
Thus $w$ satisfies the hypotheses of
Lemma~\ref{lem:hardy-critical-ground-state}. Consequently, there exist
$\beta>0$ and
\[
    V\in C([0,1])\cap C^2((0,1))
\]
such that
\[
    V(0)=V(1)=0,
    \qquad
    V(p)>0\quad\text{for }p\in(0,1),
\]
and
\[
    -V''(p)=\beta w(p)V(p),
    \qquad p\in(0,1).
\]
Since $w=G^{-2}$, this equation is equivalent to
\[
    \beta\cdot V(p)+G(p)^2 \cdot V''(p)=0,
    \qquad p\in(0,1).
\]

We now prove concavity. For $p\in(0,1)$,
\[
    V''(p)
    =
    -\frac{\beta}{G(p)^2}V(p).
\]
Because
\[
    \beta>0,\qquad G(p)>0,\qquad V(p)>0
\]
on $(0,1)$, we get
\[
    V''(p)<0,
    \qquad p\in(0,1).
\]
Thus $V$ is strictly concave on $(0,1)$ and, by continuity, concave on
$[0,1]$.

It remains to show unimodality. Since $V$ is concave, its one-sided
derivatives are monotone nonincreasing. Hence the derivative can change sign
at most once. Moreover,
\[
    V(0)=V(1)=0,
    \qquad
    V(p)>0\quad\text{for }p\in(0,1).
\]
Therefore $V$ cannot be nonincreasing near $0$, nor can it be nondecreasing
near $1$. Hence there exists $p_\ast\in[0,1]$ such that $V$ is nondecreasing
on $[0,p_\ast]$ and nonincreasing on $[p_\ast,1]$. Thus $V$ is unimodal.

Finally, let $L>0$. Since the differential equation is homogeneous and
linear,
\[
    \beta(LV)(p)+G(p)^2(LV)''(p)
    =
    L\bigl(\beta\cdot V(p)+G(p)^2 \cdot V''(p)\bigr)
    =
    0.
\]
Also,
\[
    (LV)(0)=(LV)(1)=0.
\]
Therefore $(\beta,LV)$ solves the same uniformity BVP for every $L>0$.
\end{proof}

\subsection{Proof of \cref{prop:minimax}}

\begin{proof}
Write $w(p)=G(p)^{-2}$ and set
\[
    M \;:=\; \sup_{p\in(0,1)}\frac{-V''(p)}{w(p)\,V(p)}.
\]
Then $-V''(p)\le M\,w(p)\,V(p)$ for all $p\in(0,1)$. Multiplying both sides by $V(p)>0$ and integrating over $(0,1)$ gives
\[
    \int_0^1 \bigl(-V''(p)\bigr)V(p)\,dp
    \;\le\;
    M\int_0^1 w(p)\,V(p)^2\,dp.
\]
Integrating by parts on the left and using $V(0)=V(1)=0$,
\[
    \int_0^1 |V'(p)|^2\,dp
    \;\le\;
    M\int_0^1 w(p)\,V(p)^2\,dp.
\]
Therefore
\[
    M
    \;\ge\;
    \frac{\displaystyle\int_0^1 |V'(p)|^2\,dp}
         {\displaystyle\int_0^1 w(p)\,V(p)^2\,dp}
    \;\ge\;
    \inf_{\phi\in H^1_0(0,1)\setminus\{0\}}
    \frac{\displaystyle\int_0^1 |\phi'(p)|^2\,dp}
         {\displaystyle\int_0^1 w(p)\,\phi(p)^2\,dp}
    \;=\;
    \beta,
\]
where the last equality is the Rayleigh quotient characterization of the principal eigenvalue of $-u''=\beta\,w\,u$ on $(0,1)$ with Dirichlet boundary conditions (see \cref{lem:hardy-critical-ground-state}). For the eigenfunction $V^*$, the ratio $-V^{*\prime\prime}(p)/(w(p)\,V^*(p))$ equals $\beta$ for every $p\in(0,1)$, so $M=\beta$. Conversely, equality in the Rayleigh quotient holds only for the principal eigenfunction, so $V$ must be a positive multiple of $V^*$.
\end{proof}

\subsection{Proof of \cref{cor:scale}}

If \((\beta,V)\) solves the uniformity BVP for \(G\), then
\[
    \beta\cdot V(p)+G(p)^2 \cdot V''(p)=0,
    \qquad V(0)=V(1)=0.
\]
For \(\tilde G=cG\),
\[
    c^2\beta\cdot V(p)+\tilde G(p)^2 \cdot V''(p)
    =c^2\beta\cdot V(p)+c^2G(p)^2 \cdot V''(p)
    =c^2\bigl(\beta\cdot V(p)+G(p)^2 \cdot V''(p)\bigr)
    =0.
\]
The endpoint and positivity conditions on \(V\) are unchanged, so
\((c^2\beta,V)\) solves the uniformity BVP for \(\tilde G\).

\subsection{Proof of \cref{prop:PVtoWM}}

Define
\[
    G(p)^2:=\beta\cdot\frac{V(p)}{-V''(p)},
    \qquad p\in(0,1).
\]
Since \(V(p)>0\) and \(V''(p)<0\) on \((0,1)\), $G(p)$ is
strictly positive in the interior. The assumed continuous extension of
\(-V/V''\) gives a continuous extension of \(G\) to \([0,1]\) with
\(G(0)=G(1)=0\). The assumed local Lipschitz regularity of \(-V/V''\)
implies that \(G\) is locally Lipschitz on \((0,1)\) as well. Thus \(G\)
satisfies Assumption~\ref{ass:win-martingale}(i)--(ii). Together with the
clock condition in Assumption~\ref{ass:win-martingale}(iv), the proof of
\cref{prop:separable-win-martingale} applies verbatim and implies that the
solution of
\[
    dP_t = \frac{G(P_t)}{h(t)}\,dW_t
\]
is a win-martingale; in particular, \(P_T\in\{0,1\}\) almost surely. Notice
that Assumption~\ref{ass:win-martingale}(iii) is not needed for this
win-martingale conclusion; it is used elsewhere for boundary regularity of the
uniformity BVP.

With this definition,
\[
    \beta\cdot V(p)+G(p)^2 \cdot V''(p)
    =\beta\cdot V(p)+\beta\cdot \frac{V(p)}{-V''(p)}V''(p)
    =0,
\]
so \((\beta,V)\) solves the uniformity BVP under \(G\), and $V$ is uniform under
the constructed win-martingale.

\subsection{Proof of \cref{thm:dynamic-target-loss}}

\begin{proof}
Let \(X_t:=L_tV(P_t)\) denote the marked-to-market value of the AMM position. Since
\(L\) is deterministic and differentiable, Ito's formula and the uniformity BVP
give
\[
\begin{aligned}
    dX_t
    &=\dot L_t\cdot V(P_t)\,dt+L_t\,dV(P_t) \\
    &=\left(\dot L_t-\frac{\beta}{2h(t)^2}L_t\right)V(P_t)\,dt
      +L_t\cdot V'(P_t)\cdot\frac{G(P_t)}{h(t)}\,dW_t,
\end{aligned}
\]
so the resulting wealth process satisfies
\begin{equation*}
    dW_t=dX_t-\dot L_t \cdot V(P_t)\,dt
    =-\frac{\beta}{2h(t)^2}L_t \cdot V(P_t)\,dt+dM_t,
\end{equation*}
where $\{M_t\}$ is the local martingale
\[
    M_t:=\int_0^t L_s \cdot V'(P_s) \cdot \frac{G(P_s)}{h(s)}\,dW_s.
\]
Thus, the predictable instantaneous loss is
\[
    d\LVR_t=\frac{\beta}{2h(t)^2}L_tV(P_t)\,dt,
\]
which proves part (ii).

We then compute expected wealth. Applying Ito's formula to \(V(P_t)\) yields
\[
    dV(P_t)=V'(P_t)\cdot\frac{G(P_t)}{h(t)}\,dW_t
    -\frac{\beta}{2\cdot h(t)^2}\cdot V(P_t)\,dt.
\]
Let \(a(t):=\beta/(2h(t)^2)\). We first justify the expectation identity by
localization, rather than by assuming that the stochastic integral is a true
martingale. Let \(\{\tau_n\}\) localize the stochastic integral above. For
\(t<T\), applying the stopped identity on \([0,t\wedge\tau_n]\) and taking
expectations gives
\[
    \E[V(P_{t\wedge\tau_n})]
    =
    V(p_0)
    -
    \E\!\left[\int_0^{t\wedge\tau_n} a(s)V(P_s)\,ds\right].
\]
Since \(V\) is continuous on \([0,1]\), it is bounded; hence
\(V(P_{t\wedge\tau_n})\to V(P_t)\) in expectation. The integral term is
nonnegative and increases to \(\int_0^t a(s)V(P_s)\,ds\), whose expectation is
finite because \(V\) is bounded and \(\int_0^t h(s)^{-2}\,ds<\infty\). Thus,
by monotone convergence,
\[
    m(t):=\E[V(P_t)]
    =
    V(p_0)-\int_0^t a(s)m(s)\,ds.
\]
Therefore \(m\) is absolutely continuous and satisfies
\[
    m'(t)=-\frac{\beta}{2h(t)^2}m(t) \implies m(t)=V(p_0)\exp\!\left(-\int_0^t\frac{\beta}{2h(s)^2}\,ds\right).
\]
where $m(0)=V(p_0)$. Because \(L_t\) is deterministic, substituting the liquidity schedule
yields
\[
    \E[L_t \cdot V(P_t)]
    =\frac{2w_0}{\beta}\cdot h(t)^2 \cdot D'(t).
\]
Taking expectations in the definition
\[
    W_t=c_0+L_tV(P_t)-\int_0^t \dot L_sV(P_s)\,ds
\]
and using the deterministic nature of \(L\), we obtain
\[
    \E[W_t]=c_0+L_tm(t)-\int_0^t \dot L_sm(s)\,ds.
\]
Differentiating this absolutely continuous function gives
\[
    \frac{d}{dt}\E[W_t]
    =L_t m'(t)
    =-\frac{\beta}{2h(t)^2}L_t m(t)
    =-\frac{\beta}{2h(t)^2}\E[L_tV(P_t)]
    =-w_0D'(t).
\]
Integrating from \(0\) to \(t\) and using \(D(0)=0\) gives
\[
    \E[W_t]=w_0-w_0 \cdot D(t)=w_0(1-D(t)),
\]
which proves part (i).
\end{proof}

\subsection{Proof of \cref{cor:general-withdrawal-only}}

\begin{proof}
The liquidity rule in \cref{thm:dynamic-target-loss} is a positive constant
multiple of
\[
    t\mapsto h(t)^2\exp\!\left(\int_0^t\frac{\beta}{2h(s)^2}\,ds\right)D'(t),
\]
so the given monotonicity condition implies that \(L_t\) is nonincreasing. Then,
\[
    d\left(-\int_0^t\dot L_s \cdot V(P_s)\,ds\right)=-\dot L_t \cdot V(P_t)\,dt\ge0,
\]
because \(V(P_t)\ge0\). The condition \(c_0=w_0-L_0V(p_0)\ge0\) guarantees that the liquidity provider is solvent initially.
\end{proof}

\end{document}